\documentclass[12pt]{elsarticle}

\newcommand{\rev}[1]{#1}

\usepackage{url}
\usepackage{geometry}
\usepackage{titlesec}

\titleformat*{\section}{\Large\bfseries}
\titleformat*{\subsection}{\large\bfseries}
\titleformat*{\subsubsection}{\normalsize\bfseries}

\usepackage{amsmath,amsthm,amsfonts,breqn}

\usepackage{graphicx,xcolor}

\usepackage{enumitem}
\setlist{nolistsep}

\theoremstyle{plain}
\newtheorem{theorem}{Theorem}[section]
\newtheorem{corollary}[theorem]{Corollary}
\newtheorem{lemma}[theorem]{Lemma}

\newtheorem*{theorem*}{Theorem}
\newtheorem{open}{Open question}

\theoremstyle{definition}
\newtheorem{remark}[theorem]{Remark}
\newtheorem{definition}[theorem]{Definition}
\newtheorem{example}[theorem]{Example}

\usepackage{amsfonts,amsbsy,amsmath,amssymb, amsthm}
\usepackage{verbatim} 
\usepackage{upgreek}
\usepackage{graphics}
\usepackage{color}
\usepackage{bbm}
\usepackage{mathtools}
\usepackage{xargs} 

\usepackage{algorithm}
\usepackage{algorithmic}

\numberwithin{equation}{section}
\newcommand{\discup}{\dot{\cup}}
\newcommand{\sm}{\setminus}
\renewcommand{\root}[1]{\sqrt{#1}}

\newcommand{\cS}{\mathcal{S}}

\newcommand{\Reals}{\mathbb{R}}

\theoremstyle{definition}

\newcommand{\partition}[2]{\textsc{Partitions}(#1,#2)}
\newcommand{\union}[2]{\textsc{Union}(#1,#2)}

\newcommand{\MMS}{\textsc{MMS}}
\newcommand{\mms}[4]{\MMS^{#2\text{-out-of-}#3}_#1\left(
#4
\right)}

\newcommand{\tfair}{$t$-fair}

\newcommand{\nmin}{\underline{n}}
\newcommand{\xsating}{$X$-saturating}
\newcommand{\xlsating}{$X_L$-saturating}

\usepackage{marvosym} 
\usepackage{caption} 
\usepackage{enumitem}

\graphicspath{{../graphics/}} %helpful if your graphic files are in another directory

\newcounter{problem}


\newcommand{\range}[2]{\in\{#1,\dots,#2\}}
\newcommand{\dottedline}{\hbox to 15cm{\leaders\hbox to 5pt{\hss.\hss}\hfil}}

\title{Envy-free Matchings in Bipartite Graphs and their Applications to Fair Division }

\begin{document}

\author{Elad Aigner-Horev and Erel Segal-Halevi
\\
Ariel University
\\
Kiriat Hamada 3, Ariel 40700, Israel
\\
horev.elad@gmail.com, erelsgl@gmail.com
}
\date{}

\begin{abstract}
A matching in a bipartite graph with parts $X$ and $Y$ is called envy-free, if no unmatched vertex in $X$ is a adjacent to a matched vertex in $Y$. Every perfect matching is envy-free, but envy-free matchings exist even when perfect matchings do not.

We prove that every bipartite graph has a unique partition such that all envy-free matchings are contained in one of the partition sets. Using this structural theorem, we provide a polynomial-time algorithm for finding an envy-free matching of maximum cardinality. For edge-weighted bipartite graphs, we provide a polynomial-time algorithm for finding a maximum-cardinality envy-free matching of minimum total weight. 

We show how envy-free matchings can be used in various fair division problems with either continuous resources (``cakes'') or discrete ones. In particular, we propose a symmetric algorithm for proportional cake-cutting, an algorithm for $1$-out-of-$(2n-2)$ maximin-share allocation of discrete goods, and an algorithm for $1$-out-of-$\lfloor 2n/3\rfloor$ maximin-share allocation of discrete bads among $n$ agents.

\end{abstract}

\maketitle

\newpage
\section{Introduction}
Let $G := (X \discup Y,E)$ be a bipartite graph. 
A matching $M\subseteq E$ is called
\emph{perfect} if every vertex of $X \discup Y$ is adjacent to exactly one edge of $M$;
it is called 
\emph{\xsating} if every vertex of $X$ is adjacent to exactly one edge of $M$.
This paper studies the following relaxation of \xsating{} matching (where 
$X_M$ and $Y_M$ denote the vertices of $X$ and $Y$, respectively, that are \rev{incident} to edges of $M$).
\begin{definition}
\label{def:envy-free}
\label{def:efm}
Let $G := (X \discup Y, E)$ be a bipartite graph. A matching $M \subseteq E$ is said to be {\em envy-free w.r.t. $X$} if no vertex in $X \sm X_M$ is adjacent to any vertex in $Y_M$. 
\end{definition}
One may view $X$ as a set of people and $Y$ as a set of houses, 
where a person in $X$ is adjacent to all houses in $Y$ which he or she likes.
A matching $M\subseteq E$ denotes an assignment of houses to people who like them.
Throughout the paper, all envy-free matchings are taken w.r.t. $X$. In such a matching, an unmatched person $x \in X\sm X_M$ does not envy any matched person $x'\in X_M$, because $x$ does not like any matched house $y'\in Y_M$ anyway.

If a matching $M$ is \xsating{}, then $X_M = X$, and $M$ is clearly envy-free.
A graph admitting an \xsating{} matching is called \emph{$X$-saturated}; see Figure \ref{fig:efm-examples}(a).
Many graphs are not $X$-saturated, but still admit a non-empty envy-free matching; see Figure \ref{fig:efm-examples}(b).

In contrast, in some graphs the only envy-free matching is the empty matching $\emptyset$ (which is vacuously envy-free). A natural example is
an \emph{odd path} --- 
a path with $2k+1$ vertices, for some $k \geq 1$ --- where $X$ is identified with the larger class in the bipartition; see Figure \ref{fig:efm-examples}(c).
Consider any non-empty matching in such an odd path. Traverse the path from one of its ends towards the other end, until you encounter the first matched vertex. If this vertex is in $Y$, then the previous vertex is in $X$ and it is envious. If the first matched vertex is in $X$, then it is matched to the vertex after it in $Y$, and from that point onwards, every vertex in $X$ must be matched to the vertex after it in $Y$ in order not to envy. But the last vertex of the path is in $X$ and has no vertex after it, so it is envious.

\begin{figure}
\begin{center}
\includegraphics[height=3cm,width=14cm]{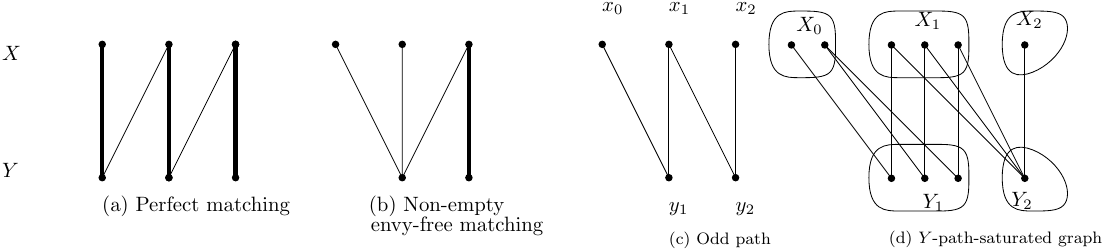}
\end{center}
\caption{\label{fig:efm-examples}
(a) An $X$-saturated graph. 
(b) A graph that is not $X$-saturated, but admits a non-empty envy-free matching, denoted in bold.
~~~
(c), (d)
Graphs in which the only envy-free matching is empty.
}
\end{figure}

The examples above invoke the following questions.
\begin{itemize}
\item What characterises the graphs that contain a non-empty envy-free matching?
\item Given a graph $G$, can an envy-free matching of maximum cardinality in $G$ be found efficiently?
\end{itemize}

\subsection{Envy-free matching and graph structure}
We answer these questions by proving a structural theorem for bipartite graphs.
We prove that in every bipartite graph, there is a unique partition of the vertices into two subsets --- ``good'' and ``bad'': the ``good'' subset is $X$-saturated (and thus contains the largest possible envy-free matching), while the ``bad'' subset has a structure similar to an odd path (and thus contains only an empty envy-free matching). The structure of this ``bad'' subset is defined formally below.
\begin{definition}
\label{def:snake}
A bipartite graph $G := (X \discup Y, E)$ is called  \emph{$Y$-path-saturated} if, for some $k\geq 1$, there exist partitions $X = X_0 \discup \cdots \discup X_k$ and $Y = Y_1 \discup \cdots \discup Y_k$ where for all $i\geq 1$:
\begin{itemize}
\item There is a perfect matching between vertices of $X_i$ and vertices of $Y_i$;
\item Every vertex in $Y_i$ is adjacent to some vertex in $X_{i-1}$.
\end{itemize}
\end{definition}
Every odd path with $|X|>|Y|$, as in Figure \ref{fig:efm-examples}(c), is $Y$-path-saturated. Figure \ref{fig:efm-examples}(d) shows another example of a $Y$-path-saturated graph; it can be seen that the structure of such a graph resembles that of an odd path.
Every $Y$-path-saturated graph is $Y$-saturated, but the opposite is not true, as shown by Figures \ref{fig:efm-examples}(a,b).%
\footnote{
Note that the empty graph is $Y$-path-saturated (where $X_i = Y_j = \emptyset$ for all $i\geq 0, j\geq 1$). 
Also note that a $Y$-path-saturated graph may have isolated vertices (vertices with degree 0) in $X$ --- such vertices are contained in $X_0$.
}
In any $Y$-path-saturated graph, the only envy-free matching is $\emptyset$. The proof is similar to the one for odd paths above; we omit it since it is implied by Theorem \ref{thm:structure}(e) below. 
Thus the $X$-saturated graphs and the $Y$-path-saturated graphs are two extreme cases:
the former contain the largest possible envy-free matching, while the latter contain only an empty envy-free matching. 
Our first result is that these two extremes are the building-blocks of \emph{all} bipartite graphs.
Below, $G[X',Y']$ denotes the subgraph of $G$ induced by the vertices $X'\discup Y'$. 

\newcommand{\theoremstructure}{
Every bipartite graph $G = (X \discup Y, E)$ admits a \emph{unique} partition $X = X_S\discup X_L$ and $Y = Y_S\discup Y_L$ satisfying the following three conditions:

(a) There are no edges between $X_S$ and $Y_L$;

(b) The subgraph $G[X_S,Y_S]$ is $Y$-path-saturated;

(c) The subgraph $G[X_L,Y_L]$
is $X$-saturated.

\noindent
Moreover, this unique partition has the following additional properties:

(d) Every \xlsating{} matching in $G[X_L,Y_L]$ is an envy-free matching in $G$.

(e) Every envy-free matching in $G$ is contained in $G[X_L,Y_L]$.
}

\begin{theorem}
\label{thm:structure}
\theoremstructure
\end{theorem}
For example, in 
Figure \ref{fig:efm-examples}(a), 
%as in any $X$-saturated graph, 
the entire graph is $G[X_L,Y_L]$, while $G[X_S,Y_S]$ is empty.
In Figure \ref{fig:efm-examples}(b), $G[X_S,Y_S]$ contains the two leftmost edges and $G[X_L,Y_L]$ contains the rightmost (bold) edge, and there is one more edge between $X_L$ and $Y_S$
(but no edges between $X_S$ and $Y_L$).
In Figures \ref{fig:efm-examples}(c,d), 
%as in any $Y$-path-saturated graph,
the entire graph is $G[X_S,Y_S]$, while $G[X_L,Y_L]$ is empty.

We call the unique partition of $G$, whose existence is guaranteed by Theorem \ref{thm:structure}, the \emph{EFM partition of $G$}.

As a corollary of Theorem \ref{thm:structure}, one gets several useful conditions on a graph $G$ admitting a non-empty envy-free matching. Two conditions are necessary and sufficient; the other is only sufficient.%
%\footnote{
~~Below, 
$N_G(X')$ denotes the neighborhood of a subset $X'\subseteq X$ in $G$, i.e.:
$
N_G(X') := \{y'\in Y: \exists x'\in X' \text{~such that~} (x',y')\in E\}
$.
%}

\begin{corollary}
\label{thm:sufficient}
A bipartite graph $G:=(X \discup Y, E)$ admits a non-empty envy-free matching ---

(a) if and only if the bipartite graph $(X\discup N_G(X), E)$ is not $Y$-path-saturated;

(b) if $|N_G(X)|\geq |X|\geq 1$;

(c) if and only if there is a subset $Y'\subseteq N_G(X)$ with $|N_G(Y')| \leq |Y'|$.
\end{corollary}

Part (a) shows that 
all the ``bad'' graphs (graphs with only an empty envy-free matching) are similar to the odd-path example --- they are all $Y$-path-saturated.

Parts (b) and (c) are similar to the condition in Hall's marriage theorem \citep{hall1935representatives}. Hall's theorem says that
if (and only if) 
$|N_G(X')|\geq |X'|$ for \emph{any} subset $X'\subseteq X$, then $G$ admits an \xsating{} matching.
The strong condition of Hall is sufficient for the strong property of having an \xsating{} matching; the weaker condition (b) is sufficient for the weaker property of having a non-empty envy-free matching.
We note that part (b) was first proved by \citet{Luria2013EnvyFree}; we present an alternative proof.

Corollary \ref{thm:sufficient}(b) can be slightly generalised to provide a lower bound on the cardinality of an envy-free matching.
\begin{corollary}
\label{thm:sufficient2}
Let $G:=(X \discup Y, E)$ be a bipartite graph with $|N_G(X)|\geq |X|\geq 1$.
If, for some integer $k\geq 1$,
every vertex in $N_G(X)$ has at least $k$ neighbors in $X$,
then $G$ admits an envy-free matching of cardinality at least $k$.
\end{corollary}

The structural theorem and its corollaries are proved in Section \ref{sec:pre}.

Once all envy-free matchings are ``captured'' within a specific subgraph $G[X_L,Y_L]$, it is easy to develop optimisation algorithms for them.
Below, the number of vertices in the smaller part of $G$ is denoted by $\nmin := \min(|X|,|Y|)$, and the number of edges by $m := |E|$.
\begin{theorem}
\label{thm:algos}
Given a bipartite graph $G = (X \discup Y, E)$,

(a) An envy-free matching of maximum cardinality in $G$ can be found in $O(m\root{\nmin})$ time. 

(b) Given an edge cost function $w : E \to \Reals_{\geq 0}$, 
an envy-free matching of minimum total cost among those of maximum cardinality can be found within 
$O(m\nmin+{\nmin}^2 \log {\nmin})$ time.
\end{theorem}

The algorithms are presented in Section \ref{sec:cardinality}.

\subsection{Envy-free matching in fair division}

A \emph{fair division problem} is a problem of allocating resources among people with different preferences, such that each person conceives his or her share as ``fair'' according to a given fairness criterion.
 
The algorithms of Theorem \ref{thm:algos} directly solve two variants of a problem known as \emph{fair house assignment}. In this problem, the resources are indivisible, each
agent must get at most a single resource,  and the fairness criterion is envy-freeness.
Part (a) solves a variant in which the goal is to maximise the number of agents assigned to a house that they like, subject to envy-freeness. 
Part (b) solves a variant in which each assignment of an agent to a house has a certain cost for society (e.g. the cost of building the house or of moving the agent to the house), and the goal is to minimise the total cost of the assignment, subject to envy-freeness and maximising the number of assigned agents.
Both parts solve variants in which it is allowed to leave some houses unallocated.

Interestingly, the same algorithms, combined with Corollary \ref{thm:sufficient},
can be used as subroutines in algorithms for various other fair division problems, both of divisible and of indivisible resources, in which \emph{all} resources must be allocated. 
Each of these problems 
requires its own notation and definitions, which are presented formally in Sections \ref{sec:app}
and \ref{sec:app-objects}. Our results are presented informally below.

For a \emph{divisible} resource (``cake''), we 
focus on a fairness criterion called \emph{proportionality}, which means that each agent must get a piece that he/she values at least a fraction $1/n$ of the total cake value \citep{Steinhaus1948Problem}. 
There are various algorithms for proportional cake division,
but most of them are not \emph{symmetric} --- the same agent might get a different value when playing first vs. playing second. 
This may lead to quarrels regarding who should play first. 
\citet{cheze2018don} presented a deterministic symmetric algorithm for proportional cake division, with an exponential run-time. He asked whether a polynomial-time algorithm exists. The following theorem answers his question; it is proved in Section \ref{sec:app}.
\begin{theorem}
\label{thm:symmetric}
There is a \emph{deterministic}, \emph{symmetric} and \emph{polynomial-time} algorithm that entirely allocates a divisible resource (``cake'') among $n$ agents such that the value of each agent is at least $1/n$ of the total cake value.
\end{theorem}

For \emph{indivisible} objects, we focus on a fairness criterion called \emph{1-out-of-$k$ maximin-share}, which means that each agent weakly prefers his or her allocated bundle over the outcome of partitioning the objects into $k$ subsets and getting the worst subset. 
\citet{procaccia2014fair} proved that, when the objects are \emph{goods} (i.e., each agent values each object at least 0),  a 1-out-of-$n$ maximin-share allocation may not exist for $n\geq 3$ agents. They asked whether a 1-out-of-$(n+1)$ maximin-share allocation exists. 
The following theorem makes a step towards an answer; it is proved in Section \ref{sec:app-objects}.
\begin{theorem}
\label{thm:mms-goods}
Given a set of indivisible goods, and $n$ agents with additive valuations, there is a protocol that partitions \emph{all} the goods among the agents, such that the value of each agent is at least the agent's 1-out-of-$(2n-2)$ maximin-share.
\end{theorem}

The same algorithm can be used when the objects are \emph{bads} (i.e., each agent values each object at most 0; such objects are also known as \emph{chores}).
\begin{theorem}
\label{thm:mms-bads}
Given a set of indivisible bads, and $n$ agents with additive valuations, there is a protocol that partitions \emph{all} the bads among the agents, such that the value of each agent is at least the agent's 1-out-of-$\lfloor 2n/3 \rfloor$ maximin-share.
\end{theorem}

Moreover, the same algorithm allows each agent to choose between other related fairness criteria, namely 
$(\ell -1)$-out-of-$(\ell n-2)$
maximin-share for any $\ell\geq 2$, or $2/3$-fraction maximin-share (see Appendix \ref{sec:mms-variants}).
The main contribution of this paper thus lies not in solving a specific fair division problem, but rather in presenting a tool --- envy-free matching --- that can be applied as a subroutine in various kinds of fair division problems.

Some extensions of the basic model and some open questions are presented in Section \ref{sec:future}.

Concepts similar to envy-free matching appeared in previous papers related to fair division, but they were hidden inside proofs of more specific algorithms \citep{Kuhn1967Games,procaccia2014fair,amanatidis2017approximation,ghodsi2018fair,bogomolnaia2020guarantees}. Appendix \ref{sec:related} presents a detailed comparison.
Presenting envy-free matching as a stand-alone graph-theoretic concept allows us to both simplify old algorithms and design new ones. 

\section{Bipartite graph structure and envy-free matchings}
\label{sec:pre}
This section proves Theorem \ref{thm:structure} and its corollaries.
%regarding the structure of bipartite graphs. 
The main technical tool used is the \emph{alternating sequence}.

\subsection{Alternating sequences}

\begin{figure}
\begin{center}
\includegraphics[height=4cm,width=7cm]{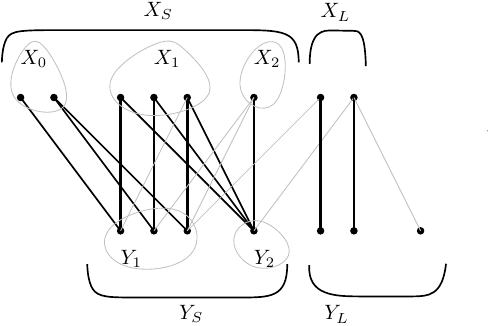}
%\hskip 1.5cm
%\includegraphics[height=3cm,width=6cm]{alternating-sequence-y}
\end{center}
\caption{
\label{fig:alternating-sequence}
A maximal $M$-alternating sequence 
$\cS(M,X_0)$ 
and the induced partitions. 
The heavy vertical edges are edges of the underlying matching $M$.
The heavy diagonal edges are 
edges of 
$E\sm M$ 
used in constructing  the alternating sequence.
The light diagonal (gray) edges are
edges of 
$E\sm M$ 
not used in the construction.
}
\end{figure}
\begin{definition}
Let $M$ be a matching in a bipartite graph $G:=(X \discup Y, E)$.
Let $X_0\subset X$ be \rev{the} subset of vertices unmatched by $M$.  An \emph{$M$-alternating sequence starting at $X_0$} is a sequence of pairwise-disjoint subsets of vertices
$
X_0 - Y_1 - X_1 - Y_2 - X_2 - \cdots
$
where for all $i\geq 1$:%
\footnote{
The $M$-alternating sequence is closely related to the \emph{$M$-alternating path} --- a sequence of vertices $x_0 - y_1 - x_1 - ...$ where each even edge is in $M$ and each odd edge is not in $M$ (or vice versa).
The difference is that the elements in an $M$-alternating sequence are \emph{subsets} of vertices rather than single vertices.
}
\begin{itemize}
\item $Y_{i} = N_{G\sm M} (X_{i-1}) ~ \sm ~ (\cup_{j<i}Y_j)$;
\footnote{
Here $G\sm M$ denotes the graph $G$ with the edges of $M$ removed.
}
\item $X_{i} = N_{M} (Y_i)$.
\end{itemize}
\end{definition}

Given $M$ and $X_0$, it is simple to construct an $M$-alternating sequence starting at $X_0$. Since the graph is finite, this construction eventually yields an empty subset --- either $X_i$ or $Y_i$ for some $i$. Denote by $\cS(M,X_0)$ the maximal $M$-alternating sequence starting at $X_0$ and ending before the first $\emptyset$.
This $\cS(M,X_0)$ induces a partition of the graph as follows (See Figure \ref{fig:alternating-sequence}):
\begin{itemize}
\item $X = X_S \discup X_L$, where $X_S := \cup_{i\geq 0} X_i$ = the vertices of $X$ participating in the sequence, and $X_L := X\sm X_S$ = the Leftover vertices. 
\item $Y = Y_S \discup Y_L$, where $Y_S := \cup_{i\geq 1} Y_i$ and $Y_L := Y\sm Y_S$.
\end{itemize}

\subsection{Alternating sequences and maximum-cardinality matchings}

When $M$ has maximum cardinality, its alternating sequences have  useful properties.%
\footnote{
The vertices of $X_L$ are exactly the vertices of $X$ that are unreachable from $X_0$ in $M$-alternating paths.
Thus $X_L$ is reminiscent of the ``unreachable'' set in the Dulmage-Mendelsohn decomposition \citep{irving2006rank,pulleyblank1995matchings}.
However, the reachability in the Dulmage-Mendelsohn decomposition is from the set of \emph{all} unsaturated vertices, while the reachability in our case is only from $X_0$ --- the set of unsaturated vertices in $X$. 
}

\begin{lemma}
\label{lem:alternating-sequence}
Let $M$ be a maximum-cardinality matching in $G=(X\discup Y,E)$ and $X_0 := X\sm X_M = $ the subset of $X$ unsaturated by $M$.
Consider the partitions $X = X_S \discup X_L$ and $Y = Y_S \discup Y_L$ induced by the maximal alternating sequence $\cS(M,X_0)$. Then:

(a) There are no edges between $X_S$ and $Y_L$;

(b) The subgraph $G[X_S,Y_S]$ is $Y$-path-saturated;

(c) The subgraph $G[X_L,Y_L]$
is $X$-saturated.
\end{lemma}
\begin{proof}
\textbf{Part (a).} By construction, the set $Y_S$ is exactly the set of  neighbors of $X_S$ in $G$.

\textbf{Part (b).} We first prove that $M[X_S,Y_S]$ --- the subset of $M$ contained in $G[X_S,Y_S]$ --- saturates $Y_S$.
Indeed, if, for some $i\geq 1$, some vertex $y_i\in Y_i$ were unmatched by $M$, then an $M$-alternating path could be traced along the edges used in the construction of $\cS(M,X_0)$, namely: $y_i - X_{i-1} - Y_{i-1} - \cdots - X_0$, where both end vertices are unmatched. 
By ``inverting'' the path, one could increase the size of the matching by one, but this contradicts the maximality of $M$. 
Hence, all vertices of $\cup_{i\geq 1} Y_i = Y_S$ are matched by $M$.
By construction, the set of their matches in $M$ is $\cup_{i\geq 1} X_i\subseteq X_S$.

This implies that the construction of $\cS(M,X_0)$ ends at the $X$ side, i.e., it ends at $X_k$ for some $k\geq 0$.
Now, the partitions $X_S = X_0 \discup \cdots \discup X_k$ and $Y_S = Y_1 \discup \cdots \discup Y_k$
satisfy the definition of a $Y$-path-saturated graph
(Definition \ref{def:snake}): for every $i\geq 1$, every vertex in $Y_i$ is adjacent to some vertex in $X_{i-1}$, 
and there is a perfect matching between $X_i$ and $Y_i$ (along edges of $M$).

\textbf{Part (c).}
We prove that $M[X_L,Y_L]$ --- the subset of $M$ contained in $G[X_L,Y_L]$ --- saturates $X_L$.
Indeed, by the lemma assumption, all vertices of $X$ unmatched by $M$ are contained in $X_0\subseteq X_S$, so all vertices of $X_L$ are matched by $M$. By construction, they must be matched to vertices not in any $Y_i$, so their matches must all lie inside $Y_L$.
\end{proof}

Note that, in the special case in which $G$ is $X$-saturated, the maximum matching $M$ saturates $X$, so $X_0$ is empty and the $M$-alternating sequence is empty. In this case, $X_L = X$ and $Y_L = Y$.
In the other extreme case, in which $G$ is an odd path, $X_0$ always contains a single vertex which is one of the two endpoints, and the $M$-alternating sequence spans the entire graph, so $X_S = X$ and $Y_S = Y$.

\subsection{Alternating sequences and envy-free matchings}

The following lemma relates the three properties (a),(b),(c) above to envy-free matchings.

\begin{lemma}
\label{lem:enclosure-efm}
Let $G = (X\discup Y, E)$, and consider any partitions
$X = X_S \discup X_L$ and $Y = Y_S \discup Y_L$ 
satisfying properties (a), (b) and (c) of Lemma \ref{lem:alternating-sequence}.
Then:

(d) Every \xlsating{} matching in $G[X_L,Y_L]$ is an envy-free matching in $G$.

(e) Every envy-free matching in $G$ is contained in $G[X_L,Y_L]$.
\end{lemma}
Note that Lemma \ref{lem:enclosure-efm} does not refer to a particular maximum matching --- it holds for any partitions of $X$ and $Y$ that satisfy the properties (a),(b),(c) above.

\begin{proof}[Proof of Lemma \ref{lem:enclosure-efm}]
~

\textbf{Part (d).}
Let $W$ be an \xlsating{} matching in $G[X_L,Y_L]$. Since $W$ saturates $X_L$, no vertex of $X_L$ is envious. 
By property (a), there are no edges between $X_S$ and $Y_L$. Since only vertices of $Y_L$ are saturated by $W$, no vertex of $X_S$ is envious. Hence, no vertex of $X$ is envious, so $W$ is an envy-free matching in $G$.

\textbf{Part (e).}
Let $W$ be any envy-free matching in $G$.
Let $X_{SW} := $ the subset of $X_S$ saturated by $W$ and $Y_{SW} := $ the subset of $Y_S$ saturated by $W$. We have to prove that both $X_{SW}$ and $Y_{SW}$ are empty.
By property (a), vertices of $X_{SW}$ can only be matched to vertices of $Y_{SW}$, so it is sufficient to prove that $Y_{SW}$ is empty.
~~
The proof is by a counting argument. Let $k_{SW} := |Y_{SW}|$ and assume by contradiction that $k_{SW} > 0$.

By property (b), the graph $G[X_S,Y_S]$ is $Y$-path-saturated; denote the partitions appearing in Definition \ref{def:snake} by
$X_S = X_0 \discup \cdots \discup X_k$ and $Y_S = Y_1 \discup \cdots \discup Y_k$.
Let $i\geq 1$ be the smallest index such that a vertex of $Y_i$ is matched by $W$, so that $Y_{SW} \subseteq \cup_{j\geq i} Y_j$.
By Definition \ref{def:snake}, all vertices of $Y_{SW}$ are perfectly matched to vertices of $\cup_{j\geq i} X_j$;
denote their matches by $X_{SW}'$. Note that 
$|X_{SW}'| = k_{SW}$.
Every vertex $x\in X_{SW}'$ is adjacent (along an edge of the perfect matching) to a vertex of $Y_{SW}$, which is saturated by $W$. To ensure that $x$ is not envious, $x$ must be saturated by $W$ too. 

Let $y'$ be a vertex in $Y_i\cap Y_{SW}$. By Definition \ref{def:snake}, it is adjacent to some vertex $x'\in X_{i-1}$. 
To ensure that $x'$ is not envious, $x'$ must be saturated by $W$ too. But $x'\not\in X_{SW}'$ since $X_{SW}'\subseteq \cup_{j\geq i} X_j$.
Hence, there must be at least $k_{SW}+1$ vertices of $X_S$ that are saturated by $W$: the $k_{SW}$ vertices of $X_{SW}'$, plus the vertex $x'$ which is not in $X_{SW}'$. 
But this is a contradiction, since vertices of $X_S$ can be matched only to vertices of $Y_S$, and only $k_{SW}$ vertices of $Y_S$ are matched by $W$.
\end{proof}

We now have all the ingredients required to prove Theorem \ref{thm:structure}, which we restate below.

\subsection{The EFM partition}

\begin{theorem*}[\ref{thm:structure}]
\theoremstructure
\end{theorem*}

\begin{proof}[Proof of Theorem \ref{thm:structure}]
Let $G = (X \discup Y, E)$ be a bipartite graph, $M$ an arbitrary maximum matching in $G$, 
and $X_S \discup X_L$ and $Y = Y_S \discup Y_L$ the partitions induced by its maximal alternating sequence.

Lemma \ref{lem:alternating-sequence} shows that these partitions satisfy properties (a), (b) and (c). 
Lemma \ref{lem:enclosure-efm} then shows that parts (d) and (e) are satisfied too.

It remains to prove that the partitions are unique, that is, do not depend on the selection of the maximum matching $M$.

Consider alternative partitions $X = X_S' \discup X_L'$ and $Y = Y_S'\discup Y_L'$ satisfying properties (a), (b) and (c).
Applying Lemma \ref{lem:enclosure-efm}(d) to the partitions $X_S\discup X_L$, $Y_S\discup Y_L$ implies that there is an envy-free matching in $G$ saturating $X_L$.
Applying Lemma \ref{lem:enclosure-efm}(e) to the partition $X_S' \discup X_L'$, $Y_S'\discup Y_L'$ implies that this matching must be contained in $G[X_L',Y_L']$; in particular, $X_L\subseteq X_L'$.
Analogous arguments imply that $X_L'\subseteq X_L$. Hence $X_L' = X_L$.
Hence also $X_S' = X_S$.
~~
Since $Y_S = N_G(X_S)$
and $Y_S' = N_G(X_S')$, we also have $Y_S' = Y_S$.
Hence also $Y_L' = Y_L$.
\end{proof}

Note how the concept of envy-free matching helped us prove the uniqueness of the partition, which is a general fact about bipartite graphs.

\subsection{Conditions for existence of envy-free matchings}

We now prove Corollary \ref{thm:sufficient}. It is simpler to prove in the following ``reverse'' formulation.
\begin{corollary}[$\equiv$ Corollary \ref{thm:sufficient}]
A bipartite graph $G:=(X \discup Y, E)$ admits only an empty envy-free matching ---

(a) if and only if the bipartite graph $(X\discup N_G(X), E)$ is $Y$-path-saturated;

(b) only if $|N_G(X)| < |X|$ or $|X| = 0$;

(c) if and only if $|N_G(Y')| > |Y'|$ for all non-empty subsets $Y'\subseteq N_G(X)$.
\end{corollary}

\begin{proof}
~\\
\textbf{Part (a).}
Consider the unique partitions $X = X_S\discup X_L$ and $Y = Y_S\discup Y_L$ that exist by Theorem \ref{thm:structure}.
Parts (d,e) of this theorem imply that $G$ admits only an empty envy-free matching iff $X_L$ is empty.
Hence it is sufficient to show that 
the graph
$(X\discup N_G(X), E)$ is $Y$-path-saturated iff
$X_L$ is empty.

If $X_L$ is empty, then $X = X_S$ and $N_G(X) = Y_S$, so the graph $(X\discup N_G(X), E)$ is $Y$-path-saturated.

Conversely, suppose $(X\discup N_G(X), E)$ is $Y$-path-saturated, and define $X_S' := X$ and $X_L' := \emptyset$ and $Y_S' := N_G(X)$ and $Y_L' := Y\setminus N_G(X)$.
Then, the partitions $X = X_S' \discup X_L'$
and $Y = Y_S' \discup Y_L'$ satisfy all three properties (a,b,c) of Theorem \ref{thm:structure}: there are no edges between $X_S'$ and $Y_L'$; the subgraph $G[X_S',Y_S']$ is $Y$-path-saturated by assumption; and the subgraph $G[X_L',Y_L']$ is vacuously $X$-saturated.
Now, the uniqueness of the partition implies that $X_L = X_L' = \emptyset$.

The next two parts follow from part (a):

\textbf{Part (b).} Every $Y$-path-saturated graph is either empty, or its $Y$ side is smaller than its $X$ side. 

\textbf{Part (c).} 
In every $Y$-path-saturated graph, every non-empty subset in the $Y$ side is contained in $\cup_{j\geq i}Y_j$ for some $i\geq 1$. It is perfectly matched to some subset of 
$\cup_{j\geq i}X_j$, and moreover, the vertices of $Y_i$ are adjacent to some vertices of $X_{i-1}$; therefore $|N_G(Y')|>|Y'|$.

Conversely, if $|N_G(Y')|>|Y'|$ for all non-empty $Y'\subseteq N_G(X)$, 
then every non-empty $Y'$ that is perfectly matched to some subset $X'\subseteq X$, must be adjacent to some vertices in $X\setminus X'$.
This implies that, in the unique EFM partition, $Y_L$ must be empty --- since otherwise it would have to be adjacent to some vertices in $X_S$, in contradiction to property (a).
\end{proof}

\begin{corollary}[$\equiv$ Corollary \ref{thm:sufficient2}]
Let $G:=(X \discup Y, E)$ be a bipartite graph with $|N_G(X)|\geq |X|\geq 1$.
If, for some integer $k\geq 1$,
every vertex in $N_G(X)$ has at least $k$ neighbors,
then $G$ admits an envy-free matching of cardinality at least $k$.
\end{corollary}
\begin{proof}
Corollary \ref{thm:sufficient}(b) implies that $G$ admits a non-empty envy-free matching.
Any matched vertex in $Y$ has at least $k$ neighbors. By envy-freeness, all these neighbors must be matched.
\end{proof}

\section{Algorithms for finding envy-free matchings}
\label{sec:cardinality}
This section 
applies the structural results of the previous section to prove Theorem \ref{thm:algos}.

The proof uses a simple algorithm (Algorithm \ref{alg:efm-partition}) for finding the unique EFM partition of a bipartite graph $G$.
The algorithm first finds an arbitrary maximum-cardinality matching in $G$; this can be done using the classic algorithm of \citet{HK73}.
\citet{RT12} show that this algorithm runs within $O(m \root{\nmin})$ time.

Then, the algorithm finds the set $X_0$ of unmatched vertices and the maximal alternating sequence $\cS(M,X_0)$.
The sets $X_S,Y_S$ are just the union of the subsets of $X,Y$ participating in the sequence, and the sets $X_L,Y_L$ are the \rev{remaining} vertices in $X,Y$. 
These can all be found in time \rev{linear in $n+m$,
since the algorithm entails to examine every vertex and scan its adjacent edges a constant number of times.} Therefore, the total run-time of Algorithm \ref{alg:efm-partition} is  $O(m \root{\nmin})$.

\subsection{Maximum cardinality envy-free matching}
\begin{proof}[Proof of Theorem \ref{thm:algos}(a)]
The theorem claims that, given a bipartite graph $G$, an envy-free matching of maximum cardinality can be found within $O(m\root{\nmin})$ time. 

This can be done simply by the following algorithm:
\begin{enumerate}
\item Find the EFM partition of $G$ using Algorithm \ref{alg:efm-partition}.
\item Return an arbitrary maximum-cardinality matching in $G[X_L,Y_L]$.
\end{enumerate}
By Theorem \ref{thm:structure}(a), the returned matching saturates $X_L$; by Theorem \ref{thm:structure}(d), it is envy-free; by Theorem \ref{thm:structure}(e), no other envy-free matching can saturate any vertex of $X_S$. Therefore, the returned matching is indeed a maximum-cardinality envy-free matching. 
\end{proof}

To save time, instead of returning an arbitrary maximum-cardinality matching in $G[X_L,Y_L]$, we can re-use the maximum-cardinality matching $M$ which is needed for computing the EFM partition, and return its subset $M[X_L,Y_L]$ --- the set of edges of the matching $M$ that link vertices of $X_L$ to vertices of $Y_L$; see Algorithm \ref{alg:cardinality}.
By Lemmas \ref{lem:alternating-sequence} and
\ref{lem:enclosure-efm}, 
this subset is an envy-free matching and it saturates $X_L$.

\begin{algorithm}[t]
\caption{Finding the EFM partition of a bipartite graph}
\label{alg:efm-partition}
\begin{algorithmic}[1]
\REQUIRE A bipartite graph $G:= (X \discup Y,E)$.  
\ENSURE The unique EFM partition of $G$ to $X_S\discup X_L = X$ and $Y_S\discup Y_L = Y$.
\STATE Find a maximum-cardinality matching $M$ in $G$.
\STATE Let $X_0 := X\sm X_M = $ the vertices of $X$ unmatched by $M$.
\STATE Compute the maximal $M$-alternating sequence $\cS(M,X_0)$.
\STATE Return $X_S := \cup_{i\geq 0} X_i$ and $X_L := X\sm X_S$ and $Y_S := \cup_{i\geq 1} Y_i$ and $Y_L := Y\sm Y_S$.
\end{algorithmic}
\end{algorithm}

\begin{algorithm}[t]
\caption{Finding an envy-free matching of maximum cardinality}
\label{alg:cardinality}
\begin{algorithmic}[1]
\REQUIRE A bipartite graph $G:= (X \discup Y,E)$.  
\ENSURE An envy-free matching of maximum cardinality in $G$. 
\STATE Find a maximum-cardinality matching $M$ in $G$.
\STATE Find the EFM partition $X=X_S\discup X_L$ and $Y=Y_S\discup Y_L$ 
using Algorithm \ref{alg:efm-partition} (steps 2--4). 
\STATE Return the sub-matching $M[X_L,Y_L]$.
\end{algorithmic}
\end{algorithm}

\begin{example}
Consider the bipartite graph in Figure \ref{fig:alternating-sequence}.
Algorithm \ref{alg:efm-partition} finds the EFM partition $X = X_S\discup X_L$ and $Y = Y_S\discup Y_L$ as indicated in the figure. 
Assuming the matching $M$ is denoted by bold vertical edges, Algorithm \ref{alg:cardinality} returns 
the set of two rightmost vertical edges, which is 
an envy-free matching of size $2$.
\end{example}

\subsection{Minimum cost envy-free matching}
\label{sub:minwgt}
\noindent
\begin{proof}[Proof of Theorem \ref{thm:algos}(b)]
The theorem assumes that the graph $G:=(X \discup Y,E)$ is endowed with an edge-cost function $w:E \to \mathbb{R}_{\geq 0}$.
The \emph{cost} of a matching $M\subseteq G$ is defined as $w(M) := \sum_{e\in M} w(e)$.
The theorem claims that 
an envy-free matching of minimum total cost among those of maximum cardinality can be found within 
$O(m\nmin+{\nmin}^2 \log {\nmin})$ time.

The proof uses Algorithm \ref{alg:minwgt}. 
Just like Algorithm \ref{alg:cardinality}, it starts by finding the unique EFM partition of $G$.
The difference is in the last step: instead of returning $M[X_L,Y_L]$, which is an arbitrary $X_L$-saturating matching, 
it returns an $X_L$-saturating matching of minimum cost.

Finding a minimum-cost maximum-cardinality matching is known as the \emph{assignment problem}.
Since $X_L$ and $Y_L$ may be of different sizes, it is an \emph{unbalanced assignment problem}.
\citet{RT12} provide a comprehensive survey of algorithms for the unbalanced assignment problem.
In particular, they show that the famous Hungarian method can be generalised to unbalanced bipartite graphs, and its run-time is $O(m\cdot \nmin + {\nmin}^2 \log{\nmin})$.

By Theorem \ref{thm:structure},
all envy-free matchings in $G$ are contained in the subgraph $G[X_L,Y_L]$, and all maximum-cardinality matchings in $G[X_L,Y_L]$ saturate $X_L$ and are therefore envy-free.
Hence, the Hungarian method on $G[X_L,Y_L]$ yields a maximum-cardinality envy-free matching of minimum cost in $G$.
Since $|X_L|\leq |Y_L|$, the run-time of the Hungarian method is $O(m\cdot |X_L| + {|X_L|}^2 \log{|X_L|})$ which is in 
$O(m\cdot \nmin + {\nmin}^2 \log{\nmin})$.
\end{proof}

\begin{algorithm}[t]
\caption{Finding a minimum-cost maximum-cardinality envy-free matching}
\label{alg:minwgt}
\begin{algorithmic}[1]
\REQUIRE A bipartite graph $G:= (X \discup Y,E)$ and a function $w : E \to \Reals_{\geq 0}$.  
\ENSURE An envy-free matching of minimum cost. 
\STATE Find the EFM partition of $G$ into $X=X_S\discup X_L$ and $Y=Y_S\discup Y_L$ using Algorithm \ref{alg:efm-partition}.
\STATE Find and return a minimum-cost maximum-cardinality matching in $G[X_L,Y_L]$. 
\end{algorithmic}
\end{algorithm}

\begin{example}
Consider again Figure \ref{fig:alternating-sequence}.
The subgraph $G[X_L,Y_L]$ contains three edges; denote them from left to right by $e_1,e_2,e_3$.
This subgraph contains two maximum-cardinality matchings: $\{e_1,e_2\}$ and $\{e_1,e_3\}$.
Both are envy-free matchings in $G$.
Algorithm \ref{alg:minwgt} returns one of them, depending on whether $w(e_1)+w(e_2)$ or $w(e_1)+w(e_3)$ is smaller.
\end{example}

\section{Application to fair division}
\label{sec:app}
\subsection{Generic fair division problem}
We consider first the following generic fair division problem.
\begin{itemize}
\item There is a set $C$, representing a resource that has to be divided among $n$ agents. 
\item For each agent $i \in[n]$ there is a measure (an additive set function) $V_i: 2^C\to \mathbb{R}$, representing the agent's valuation of different parts of the resource.
\item For each agent $i$, there is a threshold value $t_i\in \mathbb{R}$.
%; the vector of threshold values is denoted by $\mathbf{t}$.
\end{itemize}
\rev{A \emph{\tfair{} division} of $C$} is a partition of $C$ into $n$ subsets, $C = Z_1 \discup \cdots \discup Z_n$, such that 
\begin{align*}
\forall i\in[n]: V_i(Z_i)\geq t_i.
\end{align*}
The existence of \rev{a \tfair{} division} depends on the threshold values $t_i$ and on the nature of the resource $C$. For example, if $t_i > V_i(C)/n$ for all $i\in[n]$, then \rev{a \tfair{} division} obviously might not exist (e.g. when the agents' valuations are identical). 
Similarly, if $C$ contains a single (indivisible) object, and the threshold values are positive, then \rev{a \tfair{} division} does not exist. We prove that \rev{a \tfair{} division} exists whenever the threshold values are \rev{''reasonable'', in the sense defined below.}

\begin{definition}[Reasonable threshold]
\label{def:reasonable}
\rev{
Given a resource $C$, a value measure $V_i$ on $C$ and an integer $n\geq 2$, 
a real number $t_i\in \mathbb{R}$ is called \emph{a reasonable threshold for $V_i$} if:
}

(1) There exists a partition of $C$ into $C_1 \discup \cdots \discup C_n$, such that
\begin{align*}
\forall j\in[n]: V_i(C_j)\geq t_i.
\end{align*}
(Informally, $i$ can partition $C$ into $n$ subsets that are acceptable by $i$'s own standards).

(2) For every $k\in\{1,\ldots,n-1\}$, 
and every $k$ disjoint subsets $U_1,\ldots,U_k \subseteq C$, if
\begin{align*}
\forall j\in[k]: V_i(U_j) < t_i,
\end{align*}
then there exists a partition of 
$C\sm \cup_{j\in[k]} U_j$
into 
$C_1 \discup \cdots \discup C_{n-k}$, such that
\begin{align*}
\forall j\in[n-k]: V_i(C_j)\geq t_i.
\end{align*}
(Informally, if any $k$ Unacceptable subsets are given away, then $i$ can partition the remainder into  $n-k$ acceptable subsets).%
\footnote{
Condition (1) is the special case of condition (2) for $k=0$; it is presented as a different condition for the sake of clarity.
}
\end{definition}

\begin{theorem}
\label{thm:lone-divider-general}
\rev{
Consider a resource $C$ and value measures $V_1,\ldots,V_n$ on $C$.
If $t_i$ is a reasonable threshold for $V_i$ for all $i\in[n]$, then a \tfair{} division exists.
}
\end{theorem}
\begin{proof}
The proof is constructive and uses Algorithm \ref{alg:lone-divider-general}, which \rev{generalises} an algorithm of \citet{Kuhn1967Games}. It is described in detail below.

Step 1 asks some arbitrary agent $a$ to partition the resource into $|X|$ pieces that are acceptable by her own standards. In the first iteration, this is possible thanks to condition (1) in the definition of a reasonable threshold; below, we will show that this is possible in the following iterations too.

Step 2 constructs a bipartite graph where each agent is adjacent to all the pieces that are acceptable for him. 

Step 3 finds a maximum-cardinality envy-free matching in this graph. Since $a$ is adjacent to all $|X|$ pieces, $|N_G(X)|\geq |X|$, so by Corollary \ref{thm:sufficient}, a non-empty envy-free matching is found.
Each matched agent $i\in X_M$ receives a piece with a value of at least $t_i$, so the fairness condition is satisfied for these agents.

Step 4 removes the matched agents and pieces, and goes back to step 1 to handle the remaining agents.
Let $k$ be the total number of pieces allocated in all previous iterations, so that $|X\sm X_M| = n-k$.
By the definition of envy-free matching, 
for each unmatched agent $i\in X\sm X_M$, 
the value of each allocated piece is less than $t_i$.
Therefore, by condition (2) in the definition of a reasonable threshold, each remaining agent can partition the remaining resource as required in step 1.

The size of $X$ decreases by at least 1 in each iteration. Therefore, after at most $n$ iterations the algorithm ends with \rev{a \tfair{} division}.
\end{proof}

\begin{algorithm}[t]
\caption{
\label{alg:lone-divider-general}
The Lone Divider algorithm.
Based on \citet{Kuhn1967Games}.
}
\begin{algorithmic}[1]
\REQUIRE ~\\
\begin{itemize}
\item  A set $C$ representing a resource to divide.
\item  A set $X = [n]$ of agents with measures $(V_i)_{i=1}^n$ on $C$.
\item  Threshold values $(t_i)_{i=1}^n$ satisfying Conditions 1 and 2.
\end{itemize}
\ENSURE A partition $C = Z_1\discup\cdots\discup Z_n$ such that $V_i(Z_i)\geq t_i$ for all $i\in [n]$.

\dottedline{}

\STATE  
\label{step:ldg-cut}
Choose an arbitrary agent $a\in X$.
Partition $C$ into $|X|$ disjoint subsets, $(C_j)_{j\in X}$, satisfying $
V_a(C_j)\geq t_a$ for all $j\in X$.
\STATE 
\label{step:ldg-graph}
Define a bipartite graph $G$ with the agents of $X$ on the one side and the set $Y := \{C_1,\ldots,C_{|X|}\}$ on the other side. Add an edge  $(i, C_j)$ whenever $V_i(C_j)\geq t_i$. 
\STATE 
\label{step:ldg-efm}
Using Algorithm \ref{alg:cardinality},
find a maximum-cardinality envy-free matching $M$ in $G$.
Give each matched element in $Y_M$ to the agent paired to it in $X_M$.
\STATE 
\label{step:ldg-recurse}
Let $X \leftarrow X\sm X_M = $ the unallocated agents
and $C \leftarrow C \sm (\cup Y_M) = $ the unallocated resource. 
If $X\neq\emptyset$,
go back to step \ref{step:ldg-cut}.
\end{algorithmic}
\end{algorithm}

\begin{remark}
\rev{
At each iteration, 
the partition in step \ref{step:ldg-cut} requires to ask agent $a$ at most $n$ queries.
Constructing the graph in step \ref{step:ldg-graph} requires to ask each of the other agents at most $n$ queries (one query for each piece).
There are at most $n$ iterations, so each agent is asked $O(n^2)$ queries, and the total number of required queries is $O(n^3)$.
}
\end{remark}

\begin{remark}
The Lone Divider algorithm is not \emph{strategyproof} ---  agent $i$ might gain from reporting a false value measure $V_i$. This holds even when $n=2$, when Lone Divider is equivalent to cut-and-choose. See e.g. \citet{Brams1996Fair} for a discussion of this issue.
In this paper, we ignore these strategic issues and assume that the agents report their valuations truthfully.
\end{remark}

Below we apply Theorem \ref{thm:lone-divider-general}
to some specific fair division problems.

\subsection{Proportional cake-cutting}
\label{sec:app-cake}
A \emph{proportional cake-cutting} problem is a special case of the generic fair division problem, in which ---
\begin{itemize}
\item The resource $C$ is continuous; it is usually called ``cake'' and represented by the real interval $[0,1]$.
\item The value measures $(V_i)_{i \in[n]}$ are nonatomic.
\item For each agent $i$, the threshold value is $t_i = V_i(C)/n$.
\end{itemize}
\rev{These threshold values are reasonable (see Definition \ref{def:reasonable}):}
\begin{description}
\item[Condition (1)] holds thanks to the assumption that the value measures are nonatomic. For each measure $V_i$, the cake can be partitioned into $n$ subsets of equal measure, which is exactly $V_i(C)/n$.
\item[Condition (2)] holds since, if we remove from $C$ any $k$ subsets with a value smaller than $V_i(C)/n$, then the value of the remainder is at least $V_i(C) - \frac{k}{n}V_i(C) = (n-k)\cdot V_i(C)/n$; hence it can be partitioned into $n-k$ subsets of value at least $V_i(C)/n$.
\end{description}

Hence, by Theorem \ref{thm:lone-divider-general}, the Lone Divider algorithm can be used to find a proportional cake-cutting.
In fact, the Lone Divider algorithm was originally stated specifically for the proportional cake-cutting problem.
Steinhaus presented it for $3$ agents. Kuhn \citep{Kuhn1967Games} extended it to an arbitrary number of agents.
The cases $n=3$ and $n=4$ are described in detail by \citet{Brams1996Fair}[pages 31-35], and 
 the general case is described in detail by \citet{Robertson1998CakeCutting}[pages 83-87]. 
Note how the use of envy-free matchings lets us present this algorithm in a much shorter way. 

The Lone Divider algorithm requires \rev{$O(n^3)$} queries. 
Another cake-cutting algorithm, by \citet{Even1984Note}, requires only $O(n \log{n})$ queries. 
However, Lone Divider has other advantages.
One advantage is that it does not assume that all valuations are positive, or even that all valuations have the same sign: it is applicable to
a ``mixed manna'' setting,
in which each part of the cake may be positive to some agents and negative to others \citep{bogomolnaia2020guarantees,avvakumov2020equipartition}.
A second advantage of  Lone Divider is that 
can be modified to be not only fair but also \emph{symmetric}. This is explained in the following subsection.

\subsection{Symmetric algorithm for proportional cake-cutting}
This subsection proves Theorem \ref{thm:symmetric} regarding a symmetric proportional cake-cutting algorithm.
A fair division algorithm is called \emph{symmetric} if the value each agent receives depends only on the valuations of the agents, and not on the order in which the algorithm processes them. In other words, if we run the algorithm, permute the agents, and run the algorithm again, every agent has the same value in both runs.
Most cake-cutting algorithms are not symmetric. 
For example, in Algorithm \ref{alg:lone-divider-general}, the cutter in step 1 (agent $a$) always receives exactly $1/n$ of the total value, while other agents may get more than $1/n$.%
\footnote{
It is assumed here that the agents answer the queries truthfully, based on their real value measure.
}
This might make the agents quarrel over who the cutter will be.
One solution is to select the cutter uniformly at random. 
But is there a symmetric \emph{deterministic} algorithm?

\citet{Manabe2010MetaEnvyFree} presented deterministic symmetric algorithms for two and three agents. 
The case $n\geq 4$ remained open until
\citet{cheze2018don} presented a deterministic symmetric algorithm for any number of agents. 
Ch{\`e}ze mentions that the number of arithmetic operations required by his algorithm may be exponential in $n$, and asks whether there exists a deterministic symmetric algorithm in which 
the number of arithmetic operations required is polynomial in $n$. 
We answer his question in the affirmative by combining his algorithm with our Algorithm \ref{alg:minwgt} for minimum-cost envy-free matching.
The combined algorithm is shown as Algorithm 
\ref{alg:symmetric-proportional}. 
The general scheme is similar to the Lone Divider method, but there are several important changes, which are explained below.

\begin{figure}
\begin{center}
\includegraphics[height=3.5cm,width=9cm]{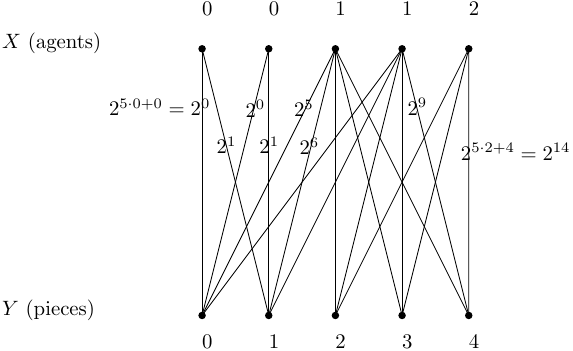}
\end{center}
\caption{\label{fig:symmetric-algorithm}
Bipartite agent-piece graph with weights determined by Algorithm \ref{alg:symmetric-proportional}.
\\~
The numbers below the pieces in $Y$ and above the agents in $X$ are their weights. 
Each piece has a unique weight, while each agent has a weight that is a 1-to-1 function of its set of neighbors.
The two leftmost agents have the same set of neighbors so they have the same weight (0); the two agents adjacent to them have the same set of neighbors so they have the same weight (1); the rightmost agent has a different set of neighbors.
\\
The numbers on the edges are their costs (due to space constraints, only some of costs are written).
}
\end{figure}
\begin{algorithm}
\caption{
\label{alg:symmetric-proportional}
A symmetric algorithm for proportional cake-cutting.
Based on \citet{cheze2018don}.
}
\begin{algorithmic}[1]
\REQUIRE A cake $C = [0,1]$, and a set $X = [n]$ of agents with nonatomic measures $(V_i)_{i=1}^n$ on $C$.
The measures are normalised such that $\forall i\in [n]: V_i(C) = n$.
\ENSURE A partition $C = Z_1\discup\cdots\discup Z_n$ such that $\forall i\in [n]: V_i(Z_i)\geq 1$.

\dottedline{}

%\STATE
%Initialize $W' := W$.
%~~~
%\COMMENT {$W'$ are the agents who have not been allocated a piece yet.}
\STATE  
\label{step:sp-mark}
Ask each agent $i\in X$ to produce \rev{a vector $z_i$ of} $n-1$ marks,
$0 < z_{i,1} < \ldots < z_{i,n-1} < 1$, which partition $C$ into $n$ sub-intervals with a value of exactly $1$ for $i$.
\STATE  
\label{step:sp-lex}
Choose \rev{from the mark vectors $\{z_i ~|~ i\in[n]\}$
the one that is smallest in lexicographic order}. Cut $C$ by these marks.
Denote the resulting pieces by $C_{1},\ldots,C_n$.
\STATE 
\label{step:sp-graph}
Define a bipartite graph $G$ with the agents of $X$ on the one side and the pieces (denoted by $Y$) on the other side. Add an edge $(i, C_j)$ whenever $V_i(C_j)\geq 1$. 
\STATE 
\label{step:sp-wgt-x}
Assign to each agent $i\in X$ a weight $w(i)\range{0}{n-1}$ that is a function of the set of neighbors $N_G(i)$, such that two agents $i,i'\in X$ have the same weight iff $N_G(i)=N_G(i')$.
\STATE 
\label{step:sp-wgt-y}
Assign to each piece $C_j\in Y$ a unique weight $w(C_j)\range{0}{n-1}$,
such that the weight of the leftmost piece (adjacent to $0$) is $0$, the weight of the next-leftmost piece is $1$, etc.
\STATE 
\label{step:sp-efm}
Assign to each edge $(i, C_j)\in E$
the cost $w(i,C_j) := 2^{n\cdot w(i) + w(C_j)}$.
Using Algorithm \ref{alg:minwgt},
find a minimum-cost maximum-cardinality envy-free matching $M$ in $G$.
\STATE 
\label{step:sp-give}
Let $X_Y\subseteq X$ be the set of all agents that are adjacent to all pieces in $Y$ (i.e., the agents $i\in X$ with $N_G(i) = Y$).
Give each agent in $X_Y$ an arbitrary piece from $N_M(X_Y)$.
\STATE 
\label{step:sp-partition}
Let $X_R := X_M \setminus X_Y$.
Partition the agents in $X_R$ to subsets according to their weight. I.e., for some $k\geq 1$, find a partition
$X_R = X_1\discup \cdots \discup X_k$,
such that $i,i'$ belong to the same $X_j$ iff $w(i)=w(i')$.
\STATE 
\label{step:sp-recurse-high}
For all $j\in[k]$, recursively divide the union of pieces in $N_M(X_j)$ among the agents in $X_j$.
\STATE 
\label{step:sp-recurse-low}
Recursively divide the union of pieces in $Y\setminus Y_M$ among the agents in $X\setminus X_M$
\rev{(by going back to step 1).}
\end{algorithmic}
\end{algorithm}

The first change from Lone Divider is that the initial partition should be decided in a way that depends only on the valuations. This is done in step \ref{step:sp-lex} using lexicographic ordering. 

For example, if Alice cuts the cake at $0.3,0.7$, Bob cuts at $0.4,0.6$ and Carl at $0.2,0.6$, then the algorithm selects Carl partition, since 
the cut-pair $(0.2,0.6)$ is lexicographically smaller than the other two cut-pairs. Hence, 
in the initial partition, the cake is cut at $0.2$ and $0.6$, so $C_1 = [0,0.2]$ and $C_2=[0.2,0.6]$ and $C_3 = [0.6,1]$.

The second change is in steps
\ref{step:sp-wgt-x}--\ref{step:sp-efm}.
Since there may be many different envy-free matchings, one of them must be selected in a way that depends only on the valuations. 
One way to select a unique envy-free matching is to assign to each edge, a cost that is a unique power of two. This guarantees that each subset of edges has a unique cost, so there is a unique minimum-cost maximum-cardinality envy-free matching.
However, symmetry requires that the edge costs themselves should depend only on the valuations.
Therefore, edges may have different costs only if they are adjacent to different pieces (since the pieces depend only on the valuations), or to agents with different valuations.
This motivates the weighting scheme in 
steps
\ref{step:sp-wgt-x}--\ref{step:sp-efm}.
An example of a graph with some edge costs is shown in 
Figure \ref{fig:symmetric-algorithm}.
Note that the length of the costs in binary is polynomial in the graph size.

In the special case that each agent has a unique set of neighbors, the agent weights are unique, the edge costs are unique powers of two, each matching has a unique cost, and thus the minimum-cost envy-free matching $M$ found in step \ref{step:sp-efm} is uniquely determined by the valuations.
In this special case, 
the algorithm can just proceed as in Algorithm \ref{alg:lone-divider-general}: give each piece in $Y_M$ to the agent matched to it in $X_M$, and recursively divide the remaining cake --- the union of pieces in $Y\sm Y_M$ --- among the remaining agents in $X\sm X_M$.

In the general case, there may be several different 
minimum-cost envy-free matchings, and 
step \ref{step:sp-efm} returns one of them, in a way that may depend on the agents' order.
Therefore, to preserve symmetry, care must be taken to ensure that the agents' values are not sensitive to the minimum-cost matching selected.
Note that all these minimum-cost matchings have the same set $X_M$ of matched agents --- it is exactly the set $X_L$ defined by the unique partition of Theorem \ref{thm:structure}.
Moreover, by the determination of edge costs, all these matchings have the same set $Y_M$ of matched pieces.
So the sets $X_M$ and $Y_M$ are uniquely determined by the valuations; only the pairing of agents in $X_M$ with pieces in $Y_M$ is not uniquely determined and must be handled in the following steps.

Consider first the set $X_Y$ defined in step \ref{step:sp-give}.
Note that it contains at least one agent --- the agent responsible to the lexicographically-smallest partition selected in step \ref{step:sp-lex}
(in Figure \ref{fig:symmetric-algorithm}, the set $X_Y$ contains the two agents with weight 1).
By envy-freeness of $M$, all agents in $X_Y$ are matched by $M$. 
By the determination of edge costs,
all minimum-cost envy-free matchings $M'$ in $G$ have the same set $N_{M'}(X_Y)$, i.e., in all these matchings, the same pieces are allocated to the agents in $X_Y$.
Since all agents in $X_Y$ value all pieces in $Y$ at exactly $1$, 
it is possible to give each agent in $X_Y$ an arbitrary piece in $N_M(X_Y)$, for example, based on the agents' indices. This arbitrary choice does not affect the value of any agent; all agents in $X_Y$ are treated symmetrically.

Consider now the sets $X_1,\ldots,X_k$ defined in step \ref{step:sp-partition}
(in Figure \ref{fig:symmetric-algorithm} there are two such sets: the set $X_1$ contains the two leftmost agents whose weight is $0$, and the set $X_2$ contains the rightmost agent whose weight is $2$).
For each $j\in[k]$, all agents in $X_j$ have the same weight, so they have the same set of neighbors.
Hence, by envy-freeness of $M$, 
if one agent in a set $X_j$ is matched by $M$, then all agents in $X_j$ must be matched by $M$ too.
By the determination of edge costs,
all minimum-cost envy-free matchings $M'$ in $G$ have the same set $N_{M'}(X_j)$, i.e., in all these matchings, the same pieces are allocated to the agents in $X_j$.
Here, it is not possible to give each agent in $X_j$ an arbitrary piece in $N_M(X_j)$, since each agent in $X_j$ may value the pieces in $N_M(X_j)$ differently.
However, by definition of the graph $G$, all agents in $X_j$ value each piece in $N_M(X_j)$ at least $1$, so they value the union of $N_M(X_j)$
at least $|X_j|$. Therefore, 
by recursively dividing the union of 
$N_M(X_j)$ among the agents in $X_j$, 
each agent in $X_j$ is guaranteed a value of at least $1$. All agents in $X_j$ are treated symmetrically. 

Finally, step \ref{step:sp-recurse-low} of Algorithm \ref{alg:symmetric-proportional}
is analogous to step \ref{step:ldg-recurse} of Algorithm \ref{alg:lone-divider-general}: 
by envy-freeness of $M$,
the agents in $X\sm X_M$ value each 
piece given away at less than $1$,
so they value the remaining cake at more than 
$|X\sm X_M|$, so the recursive call gives each of them a value of at least $1$.

From the above discussion, it follows that Algorithm \ref{alg:symmetric-proportional} is symmetric, it runs in polynomial time, and it finds a proportional cake-allocation, as claimed in Theorem \ref{thm:symmetric}.

\subsection{Other applications}
Another application of envy-free matching is found in a generalisation of cake-cutting called \emph{multi-cake cutting}. In this problem, the cake is made of $m$ pairwise-disjoint sub-cakes (``islands''), and each agent should be given a piece that overlaps at most $k$ islands, for some fixed integer $k\geq 1$. When $k<m$, it may be impossible to guarantee to each agent $1/n$ of the total value. A natural question is what fraction \emph{can} be guaranteed, as a function of $k, m, n$. 
Recently, \citet{segalhalevi2020multicake}
proved that the fraction is $\min({1\over n},{k\over m+n-1})$. The proof is constructive and uses envy-free matching in a different way than the Lone Divider algorithm.
Envy-free matchings were also applied for cutting a cake in the form of a general graph, representing e.g. a road network \citep{elkind2021graphical}.

\section{Fair allocation of discrete objects}
\label{sec:app-objects}
\newcommand{\objects}{C}

A \emph{fair object allocation} problem is a special case of the generic fair division problem, in which ---
\begin{itemize}
\item The resource $C$ is a finite set; its elements are called \emph{objects} or \emph{items}.
\item The value measures $(V_i)_{i \in[n]}$ are any additive set functions $V_i: 2^C \to \mathbb{R}$.
\end{itemize}
Objects with a positive value to all agents are usually called \emph{goods}; objects with a negative value are called \emph{bads} or \emph{chores}.

In this setting, 
a proportional allocation might not exist, i.e., 
it may be impossible to find \rev{a \tfair{} division} with threshold values $t_i = V_i(C)/n$; consider for example the case in which $C$ contains a single object.
Hence, proportionality is often relaxed to the \emph{maximin share}, which is defined below. 

\subsection{The maximin share}
For every agent $i\in [n]$ and integers $1\leq \ell\leq d$,
the \emph{$\ell$-out-of-$d$ maximin-share of $i$ from $\objects$}, denoted $\mms{i}{\ell}{d}{\objects}$, is defined as
\begin{align*}
\mms{i}{\ell}{d}{\objects} := 
~~
\max_{\mathbf{P}\in \partition{\objects}{d}}
~~
\min_{Z\in \union{\mathbf{P}}{\ell}}
~~
V_i(Z)
\end{align*}
where the maximum is over all partitions of $\objects$ into $d$ subsets, and the minimum is over all unions of $\ell$ subsets from the partition. 
Informally, $\mms{i}{\ell}{d}{\objects}$ is the largest value that agent $i$ can get by partitioning $\objects$ into $d$ piles and getting the worst $\ell$ piles.
Obviously $\mms{i}{\ell}{d}{\objects} \leq \frac{\ell}{d}V_i(\objects)$, and equality holds iff $\objects$ can be partitioned into $d$ subsets with the same value. Thus, $\mms{i}{\ell}{d}{\objects}$ can be thought of as $\frac{\ell}{d}V_i(\objects)$ ``rounded down to the nearest object''.
The maximin share with $\ell=1$ was introduced by \citet{budish2011combinatorial}.
The generalisation to arbitrary $\ell\geq 1$ was done by
\citet{Babaioff2017Competitive,babaioff2019fair}.%

The maximin-share is well-defined both for goods and for bads. For example, suppose $C$ contains three goods $o_1,o_2,o_3$, and some agent $i$ values them at $2, 3, 4$ respectively. Then $\mms{i}{1}{2}{\objects} = 4$, by the partition $\{o_1,o_2\},\{o_3\}$.
If the objects are bads and their values are $-2, -3, -4$, then $\mms{i}{1}{2}{\objects} = -5$ by the same partition.

Note that for goods $\mms{i}{1}{k}{\objects}$
is a weakly-decreasing function of $k$,
% (partitioning into a larger number of subsets decreases the value in each subset),
while for bads the opposite is true --- it is a weakly-increasing function of $k$.
%(partitioning into a larger number of subsets decreases the negative value in each subset).

\rev{
The values $t_i = \mms{i}{1}{n}{\objects}$ are particularly interesting, since they are the largest values that satisfy condition (1) for reasonable thresholds.
When $n=2$, these values satisfy condition (2) too (note that we only have to check the case $k=1$):
if one subset with value less than $\mms{i}{1}{2}{\objects}\leq V_i(\objects)/2$ is removed, then the value of the remaining objects is more than 
$V_i(\objects)/2 \geq \mms{i}{1}{2}{\objects}$.
Therefore, a \tfair{} allocation exists.
}

\citet{procaccia2014fair} prove that, for any $n\geq 3$, there might not exist a \tfair{} division with $t_i = \mms{i}{1}{n}{\objects}$. They present a \emph{multiplicative approximation} to the threshold values, 
$t_i = \gamma\cdot \mms{i}{1}{n}{\objects}$, for some fraction $\gamma\in(0,1)$. They present an algorithm attaining a \tfair{} division for a fraction $\gamma$ that equals $3/4$ for $n\in\{3,4\}$, and approaches $2/3$ as $n\to\infty$.

\rev{
An alternative approximation, suggested by \citet{budish2011combinatorial}, is $t_i =$ $\mms{i}{1}{n+1}{\objects}$. 
In contrast to the multiplicative approximation, 
the existence of a \tfair{} allocation with these thresholds depends only on the agents' rankings of the bundles, and not on the specific values assigned to them. 
In other words, an allocation that is \tfair{} with the 1-out-of-$(n+1)$ MMS thresholds remains fair even if the agents' value functions are modified, as long as the \emph{order} between the bundles' values remains the same for every agent. Therefore, we call this kind of approximation an \emph{ordinal approximation}.
}
Regarding the 1-out-of-$(n+1)$ MMS,
\citet{procaccia2014fair} say that 
\begin{quote}
``We have designed an algorithm that achieves this guarantee for the case of three players (it is already nontrivial).
Proving or disproving the existence of such allocations for a general number of players remains an open problem.''
\end{quote}
They do not present the algorithm for $n=3$.%
\footnote{Perhaps they wanted to write it in the margin but the margin was too narrow {\Smiley{}}}
Below we prove that the Lone Divider algorithm can be used to attain \rev{a \tfair{} division} with $t_i = \mms{i}{1}{2n-2}{\objects}$, which for $n=3$ coincides with $\mms{i}{1}{n+1}{\objects}$.

\subsection{Maximin-share allocation of goods}
The proof of Theorem \ref{thm:mms-goods} uses the following combinatorial lemmas.%
\footnote{
We are grateful to user bof of MathOverflow.com for the proof idea: https://mathoverflow.net/a/334754/34461
}

\begin{lemma}
\label{lem:remove-goods-a}
Let $(a_j)_{j=1}^{N}$ be real numbers 
such that for all $j\in[N]$: 
$a_j\in[0,1]$.
If $\sum_j a_j \geq A$ for some integer $A$,
then the $a_j$ can be partitioned into 
$\lceil A/2 \rceil$ subsets such that the sum of each subset is at least $1$.
\end{lemma}
\begin{proof}
Collect the $a_j$ sequentially into subsets, starting 
with $a_1$, until the sum of the current subset is at least one.
Continue constructing subsets in this way until all the $a_j$-s are arranged in subsets. Let $s+1$ be the number of constructed subsets, where the sum of the first $s$ subsets is at least $1$ and the sum of the last subset (which may be empty) is less than $1$.
Since $a_j\leq 1$, the sum of each of the first $s$ subsets is less than $2$. Therefore, the sum of all subsets is less than $2 s+1$. Hence,
$2 s+1 > A$
so $2 s \geq A$
so $s\geq \lceil A/2\rceil$, since $s$ and $A$ are integers.
\end{proof}

\begin{lemma}
\label{lem:remove-goods-b}
Let $(b_j)_{j=1}^{N}$, $(c_j)_{j=1}^{N}$  be real numbers 
such that for all $j\in[N]$: 
$b_j\in[0,c_j]$ and $c_j\geq 1$.
If $\sum_j b_j \geq (\sum_j c_j) - k$ for some integer $k\geq 0$,
then the $b_j$ can be partitioned into 
$\lceil (N-k)/2 \rceil$ subsets such that the sum of each subset is at least $1$.
\end{lemma}
Figuratively, the lemma says the following.
There are $N$ bottles of water, each of which contains at least $1$ litre.
Some water is spilled out of some of the bottles, such that the total amount spilled out is at most $k$ litres.
Then, the bottles can be grouped into  $\lceil(N-k)/2 \rceil$ subsets, such that the bottles in each subset together contain  at least $1$ litre of water.

\begin{proof}[Proof of Lemma \ref{lem:remove-goods-b}]
Let $a_j := b_j/c_j$ for all $j\in[N]$.
Then $a_j\in[0,1]$,
and 
\begin{align*}
\sum_{j=1}^N a_j 
&= \sum_{j=1}^N \frac{b_j}{c_j}
= N - \sum_{j=1}^N \frac{c_j-b_j}{c_j}
\geq  N - \sum_{j=1}^N (c_j-b_j)
&&\text{since $c_j\geq 1$}
\\
&= N - (\sum_j c_j- \sum_j  b_j)
\geq N-k
&&\text{since $\sum_j c_j- \sum_j  b_j \leq k$.}
\end{align*}
By Lemma \ref{lem:remove-goods-a},
the $a_j$ can be partitioned into $\lceil (N-k)/2 \rceil$ subsets with a sum of at least $1$.
Since $b_j = c_j\cdot a_j \geq a_j$, the sum of $b_j$ corresponding to the $a_j$ in each subset is at least $1$.
\end{proof}

We now prove Theorem \ref{thm:mms-goods}, which says that there always exists an allocation of goods among $n$ agents giving each agent $i$ a value of at least $\mms{i}{1}{2n-2}{\objects}$.

\begin{proof}[Proof of Theorem \ref{thm:mms-goods}]
We prove that  the threshold values 
$t_i = \mms{i}{1}{2n-2}{\objects}$ are reasonable, as defined in Definition \ref{def:reasonable}.
Condition (1) is obviously satisfied: by definition of MMS, each agent $i$ can partition $\objects$ into $n$ subsets worth at least
$\mms{i}{1}{n}{\objects}$, which is at least as large as 
$\mms{i}{1}{2n-2}{\objects}$.

For Condition (2), let $N := 2n-2$, and let $(C_j)_{j\in[N]}$ be a 1-out-of-$N$ MMS partition of agent $i$. 
Let $c_j := V_i(C_j) / t_i$. By definition of the MMS, $c_j \geq 1$ for all $j\in[N]$.

Suppose we remove some objects whose total value is at most $k\cdot t_i$.
Let $b_j := $ the total value remaining in $C_j$ after the removal, divided by $t_i$.
So $b_j\in[0, c_j]$, and $\sum_j b_j \geq (\sum_j c_j) - k$.
By Lemma \ref{lem:remove-goods-b}, the 
$b_j$ can be partitioned into $\lceil(N-k)/2\rceil$ subsets with a sum of at least $1$.
This corresponds to a partition of the remaining objects into $\lceil(N-k)/2\rceil$ bundles with a value of at least $t_i$.
Since 
$\lceil(N-k)/2\rceil = (n-1)-\lfloor k/2\rfloor\geq n-k$ whenever $k\geq 1$, condition (2) holds, and by Theorem \ref{thm:lone-divider-general}, the Lone Divider algorithm finds a \tfair{} division.
\end{proof}

\begin{remark}
As mentioned above, the Lone Divider algorithm \rev{requires $O(n^3)$ queries}.
However, with indivisible objects, answering each query requires agent $i$ to compute the 1-out-of-$(2n-2)$ MMS.
This requires solving an instance of the \emph{multi-way number partitioning problem}, which is known to be NP-hard.
If the number of agents and objects is sufficiently small, then the problem can be solved optimally by heuristic algorithms \citep{schreiber2018optimal}.
Otherwise, a PTAS of \citet{Woeginger1997Polynomialtime}
can be used to find in polynomial time, for each $\epsilon>0$, a partition in which the value of each part is at least $(1-\epsilon)\mms{i}{1}{2n-2}{\objects}$. Then, the Lone Divider algorithm can be executed with $t_i = (1-\epsilon)\mms{i}{1}{2n-2}{\objects}$.
\end{remark}

\begin{remark}
\label{rem:2n-2}
The Lone Divider algorithm cannot guarantee the 1-out-of-$(2n-3)$ MMS. As an example, suppose there are $4n-6$ goods. \rev{Suppose some agent Alice values some $2n-3$ goods at $1-\epsilon$ and the others at $\epsilon$, so her 1-out-of-$(2n-3)$ MMS equals $1$.
It is possible that the first divider partitions the goods such that all the $2n-3$ low-value goods are in a single bundle, and this bundle is allocated to another agent in the envy-free matching.
If, in the next round, Alice is the divider, then she cannot partition the remaining $2n-3$ high-value goods into $n-1$ bundles with a value of at least $1$.}
Recently, \citet{hosseini2021ordinal} developed a modified Lone Divider algorithm, that attains a better approximation.
\end{remark}

\subsection{Maximin-share allocation of bads}

The proof of Theorem \ref{thm:mms-bads} uses the following combinatorial lemma.

\begin{lemma}
\label{lem:remove-bads-a}
Let $(a_j)_{j=1}^{N}$ be real numbers 
such that for all $j\in[N]$: 
$a_j\in[0,1]$.
If $\sum_j a_j \leq A$ for some integer $A$,
then the $a_j$ can be partitioned into 
$2 A + 1$ subsets such that the sum of each subset is at most $1$.
\end{lemma}
\begin{proof}
Collect the $a_j$ sequentially into subsets, starting 
with $a_1$, until the sum of the current subset is at least one.
Continue constructing subsets in this way until all the $a_j$-s are arranged in subsets. Let $s+1$ be the number of constructed subsets, where the sum of the first $s$ subsets is at least $1$ and the sum of the last subset (which may be empty) is less than $1$.
The sum of all subsets is at least $s$, so $s\leq A$.
From each subset, remove the last element added to it. Since $a_j \leq 1$, we now have $2 s+1$ subsets each of which has a sum of at most $1$.
By adding empty subsets if needed, we get $2 A +1$ subsets with a sum of at most $1$.
\end{proof}

\begin{lemma}
\label{lem:remove-bads-b}
Let $(b_j)_{j=1}^{N}$, $(c_j)_{j=1}^{N}$  be real numbers 
such that for all $j\in[N]$: 
$b_j\in[0,c_j]$ and $c_j\in [0,1]$.
If $\sum_j b_j \leq (\sum_j c_j) - k$ for some integer $k\geq 0$,
then the $b_j$ can be partitioned into 
$2(N-k) + 1$ subsets such that the sum of each subset is at most $1$.
\end{lemma}
Figuratively, the lemma says the following.
There are $N$ bottles of water, each of which contains at most $1$ litre.
Some water is spilled out of some of the bottles, such that the total amount spilled out is at least $k$ litres.
Then, the bottles can be grouped into $2(N-k)+1$ subsets, such that the bottles in each subset together contain  at most $1$ litre.

\begin{proof}[Proof of Lemma \ref{lem:remove-bads-b}]
Let $a_j := b_j/c_j$ for all $j\in[N]$.
Then $a_j\in[0,1]$,
and 
\begin{align*}
\sum_{j=1}^N a_j 
&= \sum_{j=1}^N \frac{b_j}{c_j}
= N - \sum_{j=1}^N \frac{c_j-b_j}{c_j}
\leq  N - \sum_{j=1}^N (c_j-b_j)
&&\text{since $c_j\leq 1$}
\\
&= N - (\sum_j c_j- \sum_j  b_j)
\leq N-k
&&\text{since $\sum_j c_j- \sum_j  b_j \geq k$.}
\end{align*}
By Lemma \ref{lem:remove-bads-a},
the $a_j$ can be partitioned into $2 (N-k)+1$ subsets with a sum of at most $1$.
Since $b_j = c_j\cdot a_j \leq a_j$, the sum of $b_j$ corresponding to the $a_j$ in each subset is at most $1$.
\end{proof}

We now prove Theorem \ref{thm:mms-bads}, which says that there always exists an allocation of bads among $n$ agents giving each agent $i$ a value of at least $\mms{i}{1}{N}{\objects}$, where $N = \lfloor 2 n / 3 \rfloor$.

\begin{proof}[Proof of Theorem \ref{thm:mms-bads}]
We prove that  the threshold values 
$t_i = \mms{i}{1}{N}{\objects}$ are reasonable. 
Condition (1) is obviously satisfied: by definition of MMS, each agent $i$ can partition $\objects$ into $n$ subsets worth 
$\mms{i}{1}{n}{\objects}$. Since the values of all objects are negative, this value is at least 
$\mms{i}{1}{N}{\objects}$ for any $N\leq n$ (when the bads are partitioned into more subsets, the value in each subset is larger).

For condition (2),
let $(C_j)_{j\in[N]}$ be a 1-out-of-$N$ MMS partition of agent $i$. 
By definition of the MMS, $V_i(C_j) \geq t_i$.
The condition holds trivially whenever $k\leq n - N$, since in this case $n-k\geq N$, so adding $(n-k)-N$ empty bundles gives $n-k$ bundles with value at least $t_i$ (note that $t_i\leq 0$). Therefore, we assume now that $k \geq n-N+1$.

Let $c_j := V_i(C_j) / t_i$. 
Since both $V_i(C_j)$ and $t_i$ are negative, $0\leq c_j \leq 1$ for all $j\in[N]$.
Suppose we remove some objects whose total value is at most $k\cdot t_i$.
Let $b_j := $ the total value remaining in $C_j$ after the removal, divided by $t_i$.
So $b_j\in[0, c_j]$, and $\sum_j b_j \leq (\sum_j c_j) - k$
(again the sign of inequality is reversed since both quantities are negative).
By Lemma \ref{lem:remove-bads-b}, the 
$b_j$ can be partitioned into $2(N-k)+1$ subsets with a sum of at most $1$.
This corresponds to a partition of the remaining bads into bundles of value at least $t_i$.
By the definition of $N$,

\begin{align*}
&
3 N \leq 2 n
\\
\implies&
2 N +1 \leq 2 n - N +1 = n + (n-N+1)
\\
\implies&
2 N +1 \leq n + k && \text{whenever $k\geq n-N+1$}
\\
\implies&
2 (N -k) + 1 \leq n - k.
\end{align*}
By adding empty bundles if needed, agent $i$ can partition the remaining bads into $n - k$ bundles with a value of at least $t_i$.  
Therefore, condition (2) holds, and by Theorem \ref{thm:lone-divider-general}, the Lone Divider algorithm finds \rev{a \tfair{} division}.
\end{proof}

\begin{remark}
\label{rem:2n/3}
Suppose $n$ is divisible by $3$ and $N=2 n / 3 $. Then the Lone Divider algorithm cannot guarantee the 1-out-of-$(N+1)$ MMS.
Suppose the bads' values for Alice are
\begin{itemize}
\item 
$N+1$ big bads with value  $-1/2-1/2V$ each, for $V = N+\frac{n}{3}+1$;
\item 
$N+1$ sets of 
$V-1$ small bads with value $-1/2 V$ each. So the total value of each set is $-1/2+1/2V$, and the total number of small bads is $(N+1)(V-1) = N V +  \frac{n}{3}$.
\end{itemize}
The bads can be grouped into $N+1$ sets with value $-1$, so $\mms{A}{1}{N+1}{\objects} = -1$. 
But it is possible that the first envy-free matching allocates $k = n-N$ unacceptable bundles, each of which contains $2 V+1$ small bads (with total value $-1-1/2 V$);
note that the total number of small bads allocated is 
$(n-N)(2V+1) = \frac{n}{3}(2V+1) = N V + \frac{n}{3}$.
Then $N$ agents remain,
and there are $N+1$ big bads that Alice cannot partition into $N$ bundles with value at least $-1$.
\end{remark}

\subsection{Other maximin-share guarantees}
The Lone Divider algorithm can find \rev{a \tfair{} division} with various other threshold values (see Appendix \ref{sec:mms-variants}). 
For example, it can guarantee to each agent his
$2$-out-of-$(3n-2)$ MMS.
For some agents, this guarantee may be better than $1$-out-of-$(2n-2)$ MMS. For example, with $n=3$, if $\objects$ contains $7$ objects and agent $i$ values all of them at $1$, then
$\mms{i}{2}{7}{\objects}=2$ 
while $\mms{i}{1}{4}{\objects}=1$. 
More generally, for every integer $\ell\geq 2$, it is possible to find \rev{a \tfair{} division} with  
$t_i = \mms{i}{\ell-1}{\ell n-2}{\objects}$. 
The algorithm can also guarantee a multiplicative approximation of  $t_i = {2 n \over 3n - 1}\cdot \mms{i}{1}{n}{\objects} > {2 \over 3}\mms{i}{1}{n}{\objects} $.
\rev{
For proving all these variants, it is sufficient to prove that the threshold vectors are reasonable, using combinatorial lemmas analogous to Lemmas \ref{lem:remove-goods-a} and \ref{lem:remove-goods-b}}; see 
Appendix \ref{sec:mms-variants} for details.

Algorithm \ref{alg:lone-divider-general} can even make different guarantees to different agents. 
For example, it is possible to set 
$t_1 = \mms{1}{1}{2n-2}{\objects}$ and
$t_2 = \mms{2}{2}{3n-2}{\objects}$ 
and
$t_3 = {2\over 3}\mms{3}{1}{n}{\objects}$.
Finally, if some agents are computationally-bounded, and cannot calculate their MMS partition exactly,
they can use an approximation algorithm like that of \citet{Woeginger1997Polynomialtime} 
to calculate an approximate MMS partition --- a partition in which the value of each bundle is at least $(1-\epsilon)\mms{i}{1}{2n-2}{\objects}$, for some $\epsilon>0$.
They can then participate in Algorithm \ref{alg:lone-divider-general}
with 
$t_i = (1-\epsilon)\mms{i}{1}{2n-2}{\objects}$.
The algorithm then guarantees to these agents a value of at least their approximate MMS.
This does not affect the guarantee to computationally-unbounded agents, who are still guaranteed at least their exact MMS.

While there are now algorithms that attain better multiplicative approximation factors for goods
 \cite{amanatidis2017approximation,barman2017approximation,ghodsi2018fair,garg2019approximating} and for bads
 \cite{barman2017approximation,huang2019algorithmic},
 it may be useful to have a simple algorithm that allows each agent to choose between a multiplicative and various ordinal approximations.%
\footnote{
Recently, \citet{bogomolnaia2020guarantees} have shown that an algorithm very similar to Algorithm \ref{alg:lone-divider-general} attains a different approximate-fairness notion that they call ``Pro1'' (proportionality up to at most one object).
}

In general, the Lone Divider method cannot guarantee a 1-out-of-$(n+1)$ MMS allocation (see Remark \ref{rem:2n-2}).  
Corollary \ref{thm:sufficient2} can help to identify special cases in which such allocations do exist.
First, recall that, if all $n$ agents have the same valuation function, then by definition a 1-out-of-$n$ MMS allocation exists.  The same is true if all $n$ agents have the same MMS partition (even if their valuations are different).
Moreover, the same is true even if only $n-1$ agents have the same MMS partition, since then it is possible to let the $n$-th agent pick a bundle and divide the remaining $n-1$ bundles among the remaining $n-1$ agents. The following theorem generalizes this observation.
\begin{theorem}
\label{thm:nl}
Let $\ell\geq 1$ be an integer.
Given a set of indivisible goods, 
if there exists a partition $P$ of the goods into $n$ bundles,
in which $n-\ell$ agents value each bundle at least as their 1-out-of-$(n+\ell-1)$ MMS, 
then there exists a 1-out-of-$(n+\ell-1)$ MMS allocation.
In particular, existence \rev{is ensured} if $n-\ell$ agents have identical valuations.
\end{theorem} 

\begin{proof}[Proof of Theorem \ref{thm:nl}]
Let $N := n+\ell-1$.
Apply the Lone Divider algorithm with $t_i = \mms{i}{1}{N}{C}$ for all $i\in[n]$, using in Step 1 the partition $P$ from the theorem statement. 
In Step 2, by assumption, each bundle in $Y$ is adjacent to at least $n-\ell$ agents in $X$. Hence, by Corollary \ref{thm:sufficient2}, the number of matched agents is at least $n-\ell$. Denote this number by $k$.
Lemma \ref{lem:remove-goods-b} implies that each of the remaining $n-k$ agents can partition the remaining objects into 
$\lceil (N-k) / 2 \rceil$ bundles worth at least $t_i$.
The assumption $k\geq n-\ell$ implies that 

\begin{align*}
\lceil (N-k) / 2 \rceil &= \lceil (n+\ell-1-k) / 2 \rceil
\\
& \geq 
\lceil (n+(n-k)-1-k) / 2 \rceil
\\
& =
\lceil (2n-2k-1) / 2 \rceil
\\
& =
n-k.
\end{align*}
Hence, we can proceed with Algorithm \ref{alg:lone-divider-general} and get \rev{a \tfair{} division.}
\end{proof}

As an example, for $n=4$, a 1-out-of-$5$ MMS allocation exists 
whenever there exists a partition in which some two agents value each bundle by at least their 1-out-of-$5$ MMS; particularly, when some two agents have identical valuations.
The general case remains open.

\section{Extensions and open problems}
\label{sec:future}

\subsection{Symmetric envy and non-bipartite graphs}
Our definition of an envy-free matching is asymmetric in that it considers the envy of vertices in $X$ only.
For example, the odd path in Figure \ref{fig:efm-examples}(c) has a non-empty EFM w.r.t. $Y$ but not w.r.t. $X$.

One can define a matching as \emph{symmetric-envy-free} if any unmatched vertex in $G$ is not adjacent to any matched vertex in $G$.
This definition extends naturally to non-bipartite graphs.

With this symmetric definition, the algorithmic problems studied here become much easier (and less interesting). 
Suppose first that $G = (V,E)$ is a connected graph (bipartite or not). 
If some matching in $G$ saturates 
some vertex $v\in V$ but does not saturate some other vertex $v'\in V$, 
then on the path between $v$ and $v'$, at least one vertex is envious. 
Therefore, a symmetric envy-free matching saturates either all vertices or no vertices.
Hence, a connected graph admits a
non-empty symmetric-envy-free matching if and only if it admits a perfect matching. 

Therefore, an arbitrary graph admits a non-empty symmetric-envy-free matching if and only if it has a connected component admitting a perfect matching. 
A maximum cardinality (minimum cost) symmetric-envy-free matching is just the union of all perfect matchings (of minimum cost) of such connected components.

\subsection{Star matchings}
\label{sub:star}
The envy-freeness concept can be generalised from a matching (a set of vertex-disjoint edges)
to an \emph{$r$-star matching} --- 
a set of vertex-disjoint copies of the 
the star $K_{1,r}$, where the star center is in $X$ and the star leaves are in $Y$.
An \emph{envy-free $r$-star matching} is then an $r$-star matching in which every vertex in $X$ that is not matched (as a center), is disconnected from any vertex in $Y$ that is matched (as a leaf).
Our results can be easily generalised to $r$-star matchings. For example, 
the following theorem generalises Theorem \ref{thm:algos}(a) and Corollary \ref{thm:sufficient}(b).
\begin{theorem}
\label{thm:star}
For every integer $r\geq 1$,

(a) 
There is a polynomial-time algorithm that, given any bipartite graph $G=(X \discup Y, E)$, finds a maximum-cardinality envy-free $r$-star matching in $G$.

(b) 
If $|N_G(X)|\geq r|X|\geq 1$, 
then $G$ admits a non-empty envy-free $r$-star matching.
\end{theorem}
\begin{proof}
Given $G$, 
construct an auxiliary bipartite graph $G' := (X' \discup Y,E')$, 
where $X'$ has $r$ clones of every vertex in $X$, 
and $E'$ has an edge from each clone $v^x$ of $x \in X$ to every vertex $y\in N_G(x)$.
For $x \in X$, let $v_1^x,\ldots,v_r^x$ denote the $r$ clones of $x$ in $X'$. 
There is a many-to-one correspondence between 
envy-free matchings in $G'$
and
envy-free $r$-star matchings in $G$:

(1) Consider any envy-free matching $W'$ in $G'$. For all $x\in X$, 
consider 
the subgraph 
\begin{align*}
G'[\{v_1^x,\ldots,v_r^x\}, N_G(x)].
\end{align*}
This subgraph forms a complete bipartite graph. Hence, if $(v_i^x, y) \in W'$ for some $i \in [r]$, then $W'$ being envy-free in $G'$ means that $W'$ must saturate all of $\{v_1^x,\ldots,v_r^x\}$. All these vertices can only be paired to vertices of $N_G(x)$; all edges thus used by $W'$ (in $G'$) correspond to different edges in $G$ whose one end is $x$ (the ends in $N_G(x)$ are distinct). 
Collapsing every set $\{v_1^x,\ldots,v_r^x\}$ saturated by $W'$ back to its origin $x$ in $G$ gives an envy-free $r$-star matching in $G$. 

(2) Given an envy-free $r$-star matching $W$ in $G$, 
create a matching $W'$ in $G'$ by connecting, for each saturated vertex $x\in X$, 
each clone $v_i^x$ of $x$ to one of the $r$ vertices in $Y$ matched to $x$ in $W$.
Note that, for each saturated vertex $x\in X$, there are $r!$ ways to connect the clones of $x$ to its neighbors, so there are many different matchings $W'$ corresponding to $W$.
However, all such matchings have the same cardinality, and every such matching is envy-free in $G'$: 
for every vertex $x$ that is saturated by $W$, all its clones are saturated by $W'$ and thus are not envious;
for every vertex $x$ that is unmatched by $W$, all its clones are not adjacent to any matched vertex in $Y$, and thus are not envious either.

We now use the above correspondence for proving the two claims in the theorem.

(a) The size (number of edges) of the matchings $W'$ in $G'$ is exactly $r$ times the size (number of stars) of the corresponding matching $W$ in $G$. Hence, applying Algorithm \ref{alg:cardinality} to $G'$ yields a maximum-cardinality $r$-star matching in $G$.

(b) If $|N_G(X)|\geq r|X|\geq 1$, then 
$|N_{G'}(X')|\geq |X'|\geq r\geq 1$, so 
$G'$ satisfies the premise of
Corollary \ref{thm:sufficient}(b)
and thus admits a non-empty envy-free matching $W'$. It corresponds to an envy-free $r$-star matching $W$ in $G$.
\end{proof}

In an $r$-star matching, each vertex in $X$ is matched to either $0$ or $r$ vertices in $Y$.
One can also consider allocation problems in which each vertex in $X$ may be connected to \emph{any number} in $\{0,\ldots,r\}$ of vertices in $Y$ (but each vertex in $Y$ may still be connected to at most one vertex in $X$).
A many-to-one matching $M\subseteq E$ is called 
\emph{envy-free} if for every two vertices
$x_1, x_2\in X$, 
the number of neighbors of $x_1$ matched to $x_1$
is at least as large as the number of neighbors of $x_1$ matched to $x_2$.
This definition reduces to Definition \ref{def:efm} when $r=1$.

When $r=\infty$, the problem of finding an envy-free many-to-one matching is equivalent to the problem of \emph{fair allocation with binary additive valuations}. $Y$ is a set of discrete goods, and $X$ is a set of agents. Each agent values each object at either $0$ or $1$, and values each set of objects as the sum of the values of its elements.
The goal is to allocate the objects among the agents such that each agent values its own bundle at least as much as the bundle of any other agent.
\citet{Aziz2015Fair} proved that deciding whether an envy-free allocation of all objects in $Y$ exists is NP-complete (remark after Theorem 11; the same result was proved in a different way by \citet{hosseini2019fair} at Proposition 3).
Therefore, the problem of finding an envy-free one-to-many matching of maximum cardinality is NP-hard.
However, the reductions 
consider only allocations in which all objects are allocated --- they do not allow partial allocations.%
\footnote{
For example, in the reduction of 
\citet{hosseini2019fair},
Property 2 does not necessarily hold for partial allocations: it is possible that each edge-agent receives a single edge-good, each dummy-agent receives nothing, and the vertex-goods remain unallocated. This is a non-empty envy-free allocation that does not correspond to an equitable coloring of $G$.
}
Therefore, the following problem remains open.

\begin{open}
Is there a polynomial-time algorithm for deciding whether a given bipartite graph admits a non-empty envy-free one-to-many matching?
\end{open}

\subsection{Maximum value envy-free matching}
Suppose that the edge weights are interpreted as \emph{values} rather than costs, and thus one is interested in finding an envy-free matching of \emph{maximum} total value.

The unbalanced Hungarian method can be easily adapted to find an \xsating{} matching of maximum value \citep{RT12}.
Hence, Algorithm \ref{alg:minwgt} can be 
adapted to find a maximum-cardinality envy-free matching of maximum value. 
The following lemma shows that, whenever all values are non-negative,
the maximum value of \emph{any} envy-free matching is always attained by some \emph{maximum-cardinality} envy-free matching.%
\footnote{
Note that the above does not hold for arbitrary matchings.
}

\begin{lemma}\label{lem:x-saturating}
Let $G:=(X \discup Y,E)$ be a bipartite graph.
Let $W$ be an envy-free matching in $G$. 
Then, $W$ is contained in some maximum-cardinality envy-free matching in $G$.
\end{lemma}
\begin{proof}
Let $M$ be a maximum-cardinality envy-free matching in $G$.
By Theorem \ref{thm:structure}, this $M$ is contained in $G[X_L,Y_L]$ and saturates $X_L$.
For each vertex $x\in X_L$, denote by $N_M(x)$ the vertex in $Y_L$ matched to it by $M$.
Let $M' := \{(x,N_M(x)): x\in X_L, x\text{ is unsaturated by }W\}$,

Envy-freeness of $W$ implies that,
for any vertex $x\in X_L$ unsaturated by $W$, $N_M(x)$ must be unsaturated by $W$ too.
Hence, $W\cup M'$ is a matching. It saturates $X_L$, so by Theorem \ref{thm:structure} it is a maximum-cardinality envy-free matching in $G$.
\end{proof}
Hence, by taking 
Algorithm \ref{alg:minwgt} and replacing ``minimum-cost'' by ``maximum-value'', 
one gets an algorithm for finding a maximum-value envy-free matching.

However, this algorithm is meaningful only when $w(x,y)$ represent the value of the pairing $(x,y)$ to ``society'' as a whole (or to the social planner), since it appears in the maximisation objective but not in the envy definition.
An alternative interpretation is that $w(x,y)$ represents the subjective value of $y$ to $x$. This interpretation leads to a different definition of envy-freeness. Given a function $w$ on the edges and a matching $M$, 
define:

\begin{align*}
w(x, M) = 
\begin{cases}
w(x,y) & \text{If $(x,y)\in M$;}
\\
0 & \text{If $x$ is unmatched by $M$.}
\end{cases}
\end{align*}
A matching $M\subseteq E$ is called \emph{$w$-envy-free} if, 
for every vertex $x\in X$ and every matched vertex $y'\in Y_M$: $w(x,M) \geq w(x,y')$, i.e, every agent in $X$ weakly prefers his or her own house (if any) to any house assigned to another agent.
This definition reduces to Definition \ref{def:efm} when all edges have the same weight.
The problem of finding a $w$-envy-free matching was studied by several authors in parallel to the present work:
\begin{itemize}
\item \citet{gan2019envy} present a polynomial-time algorithm for finding an $X$-saturating $w$-envy-free matching, if and only if such a matching exists.
\item 
\citet{beynier2019local}
consider a similar problem in a more complex setting where agents are located on a network, and each agent only envies his or her neighbors in the network.
\item 
\citet{kamiyama2021complexity} study the problem of finding an $X$-saturating matching that is not necessarily $w$-envy-free, but it maximizes the number of vertices of $X$ for which the $w$-envy-free condition is satisfied. They prove that this problem is NP-hard even for binary weights (all weights are either $0$ or $1$). Moreover, for general weights, the problem is hard to approximate under some common complexity-theoretic assumptions. 
\end{itemize}

\rev{
In contrast to our work, 
these three works do not consider partial matchings (matchings that do not necessarily saturate $X$).
To illustrate the difference between the settings, suppose $|X|=|Y|=n$, and
\begin{align*}
w(x_i,y_j) = 
\begin{cases}
2 & 1\leq i\leq n, j=n;
\\
1 & i=j < n.
\end{cases}
\end{align*}
In the unique $X$-saturating matching, $x_i$ is matched to $y_i$ for all $i\in[n]$, and $w$-envy-freeness is satisfied only for $x_n$.
However, the matching in which $x_i$ is matched to $y_i$ for all $i\in[n-1]$ and $x_n$ remains unmatched is a partial $w$-envy-free matching of size $n-1$.
% ($x_n$ is not envious since it is adjacent only to $y_n$, which is unmatched).
}

\begin{open}
Is there a polynomial-time algorithm that, for any value function $w$, finds a partial $w$-envy-free matching of maximum cardinality? Of maximum value?
\end{open}

\ifdefined\RANKMAXIMAL
\begin{remark}
\citet{irving2006rank} present an algorithm for finding a \emph{rank-maximal} matching --- a matching that maximises the number of agents who are matched to their first choice, then the number of agents who are matched to their second choice, etc. 
A matching that is rank-maximal among the set of all envy-free matchings can be found by an algorithm similar to Algorithm \ref{alg:minwgt}: in the last step, instead of solving the unbalanced assignment problem, apply the algorithm of \citet{irving2006rank} to $G[X_L,Y_L]$. 
\end{remark}
\fi

\subsection{Relaxations of envy-free matching}
Since non-empty envy-free matchings might not exist, one may be interested in relaxations. 
For example, given a real $\alpha >0$, 
a matching $M$ is called \emph{$\alpha$-fraction envy-free} if for every unmatched $x \in X\sm X_M$:
$
|N_G(x) \cap Y_M| \leq \alpha |N_G(x)|
$.
That is, agents unsaturated by $M$ are willing to ``tolerate'' at most an $\alpha$-fraction of their acceptable houses being assigned to someone else.
Alternatively, given an integer $c\geq 0$, 
$M$ is called \emph{$c$-additive envy-free}
if 
for every unmatched $x \in X\sm X_M$:
$
|N_G(x) \cap Y_M| \leq c
$. 

\begin{open}
Is there a polynomial-time algorithm for finding 
a maximum-cardinality matching
among the approximate-envy-free matchings, for any of the above approximation notions?
\end{open}
%The complete-graph example in Figure \ref{fig:efm-examples}(c) indicates that such matchings might not exist for any $\alpha<1$ and $c\geq 0$. Still one may look for a maximum-cardinality matching satisfying these notions.

From a probabilistic perspective, it may be interesting to calculate the probability that a non-empty envy-free matching exists in a random graph. This is related to the problem of calculating the probability that an envy-free allocation exists, which has recently been studied by e.g. \citet{Dickerson2014Computational} and
 \citet{manurangsi2019envy}.

\section{Acknowledgments}
Erel acknowledges Zur Luria \citep{Luria2013EnvyFree}, 
who first provided an existential proof to Corollary \ref{thm:sufficient}(b), as well as instructive answers by Yuval Filmus, Thomas Klimpel and bof in MathOverflow.com,
%https://mathoverflow.net/a/334754/34461
Max, Vincent Tam and Elmex80s in MathStackExchange.com,
%https://math.stackexchange.com/q/3951772/29780
and helpful comments by anonymous referees to the WTAF 2019 workshop, the EC 2020 conference, and the Information Sciences journal.
This research is partly supported by Israel Science Foundation grant 712/20.

\newpage
\appendix

\renewcommand{\thesection}{\Alph{section}}

\section{Variants of Maximin Share Fairness}
\label{sec:mms-variants}
This appendix shows various fairness guarantees that can be attained by the Lone Divider algorithm when allocating discrete goods. 

\subsection{Ordinal approximation}
The first fairness guarantee uses several lemmas.

\begin{lemma}
\label{lem:remove-goods-aa}
Let $\ell \geq 2$ be an integer.
Let $(a_j)_{j=1}^{N}$ be real numbers 
such that for all $j\in[N]$: 
$a_j\in[0,1]$.
If $\sum_j a_j \geq A$ for some integer $A$,
then the $a_j$ can be partitioned into 
$\lfloor (A+1)/\ell \rfloor$ subsets such that the sum of each subset is at least $\ell-1$.
\end{lemma}
Note that Lemma \ref{lem:remove-goods-a} corresponds to the special case  $\ell=2$.
\begin{proof}
Collect the $a_j$ sequentially into subsets, starting 
with $a_1$, until the sum of the current subset is at least $\ell-1$.
Continue constructing subsets in this way until all the $a_j$-s are arranged in subsets. Let $s+1$ be the number of constructed subsets, where the sum of the first $s$ subsets is at least $\ell-1$ and the sum of the last subset (which may be empty) is less than $\ell-1$.
Since $a_j\leq 1$, the sum of each of the first $s$ subsets is less than $\ell$. Therefore, the sum of all subsets is less than $\ell s+\ell-1$. Hence,
$\ell s+\ell-1 > A$
so $\ell(s+1) > A+1$
so $s + 1 > \lfloor (A+1)/\ell \rfloor$
so $s\geq \lfloor (A+1)/\ell \rfloor$.
\end{proof}

\begin{lemma}
\label{lem:products}
Let $L \geq 1$ be an integer.
Let $(d_j)_{j=1}^{N}$, $(c_j)_{j=1}^{N}$  be real numbers such that for all $j\in[N]$:
$d_j\in[0,1]$,
and the sum of every $L$-tuple of $c_j$ is at least $L$.
If $\sum_j d_j\geq L$, then 
$\sum_j c_j\cdot d_j \geq \sum_j d_j$
(in particular
$\sum_j d_j\geq L$
implies 
$\sum_j c_j\cdot d_j \geq L$
and
$\sum_j d_j>L$
implies 
$\sum_j c_j\cdot d_j >L$
).

\end{lemma}
Note that in the special case  $L=1$, $c_j\geq 1$ for all $j\in[N]$, so the claim is trivial.

\newcommand{\dgeq}{D_{\geq}}
\newcommand{\deq}{D_{=}}
\begin{proof}
%Assume w.l.o.g. that $c_1\geq \cdots \geq c_N$.
Let $\dgeq \subset [0,1]^N$ be the set of vectors $\mathbf{d}$ satisfying the lemma condition, i.e., $\sum_j d_j\geq L$.
Let $\deq\subset \dgeq$ be those vectors satisfying $\sum_j d_j = L$.
The lemma can be stated as a dot product: $\mathbf{c}\cdot \mathbf{d}\geq L$.
We first show that it holds for all $\mathbf{d}\in \deq$ and then for all $\mathbf{d}\in \dgeq$.

Consider a vector $\mathbf{d}\in \deq$.
If it has only integer coordinates, then it must have exactly $L$ ones and $N-L$ zeros, so the sum $\mathbf{c}\cdot \mathbf{d} = \sum_j c_j \cdot d_j$ contains exactly $L$ elements from $\mathbf{c}$. By assumption, their sum is at least $L$, so the lemma holds.

Otherwise, $\mathbf{d}$ has a non-integer coordinate, say $d_{j1} \in(0,1)$. Since the sum of coordinates is an integer --- it must have another non-integer coordinate, say $d_{j2} \in (0,1)$.
Given $\epsilon>0$, 
define a \emph{$(j_1,j_2,\epsilon)$-shift} of $\mathbf{d}$ as a vector $\mathbf{d}'$ given by
\begin{align*}
d'_{j} = 
\begin{cases}
d_{j1}-\epsilon & j=j_1
\\
d_{j2}+\epsilon & j=j_2
\\
d_j & \text{otherwise}
\end{cases}
\end{align*}
Choosing $\epsilon = \min(d_{j1}, 1-d_{j2})$
guarantees that $\mathbf{d}'\in D_=$ and it has fewer non-integer coordinates than $\mathbf{d}$
(either $d'_{j1}=0$ and $d'_{j2}=d_{j1}+d_{j2}$, or $d'_{j2}=1$ and $d'_{j1} = d_{j1}+d_{j2}-1$).
Choosing $j_1,j_2$ such that $c_{j1}\geq c_{j2}$ guarantees that each such shift decreases the dot product,
\begin{align*}
\mathbf{c}\cdot \mathbf{d}'
=
\mathbf{c}\cdot \mathbf{d}
-
c_{j1}\cdot \epsilon
+
c_{j2}\cdot \epsilon
\leq
\mathbf{c}\cdot \mathbf{d}
\end{align*}
There is a finite sequence of $(j_1,j_2,\epsilon)$-shifts culminating in a vector $\mathbf{d}''\in D_=$ with only integer coordinates. Therefore,
\begin{align*}
\mathbf{c}\cdot \mathbf{d}
\geq
\mathbf{c}\cdot \mathbf{d}''
\geq L,
\end{align*}
so the claim holds for all $\mathbf{d}\in D_=$.

For a vector $\mathbf{d}\in \dgeq$,
Let $\mathbf{d'} := \mathbf{d}\cdot \frac{L}{\sum_j d_j}$; by assumption $\mathbf{d'}\in \deq$, so $\mathbf{c}\cdot \mathbf{d'}\geq L$. Now,
\begin{align*}
\mathbf{c}\cdot \mathbf{d} 
= \frac{\sum_j d_j}{L} \mathbf{c}\cdot \mathbf{d'} 
= \sum_j d_j\cdot \frac{\mathbf{c}\cdot \mathbf{d'}}{L}  
\geq \sum_j d_j \geq L,
\end{align*}
so the claim holds for all $\mathbf{d}\in \dgeq$ too.
\end{proof}

\begin{lemma}
\label{lem:remove-goods-bb}
Let $\ell \geq 2$ be an integer.
Let $(b_j)_{j=1}^{N}$, $(c_j)_{j=1}^{N}$  be real numbers 
such that for all $j\in[N]$: 
$b_j\in[0,c_j]$,
and the sum of every $(\ell-1)$-tuple of $c_j$ is at least $\ell-1$.
If $\sum_j b_j \geq (\sum_j c_j) - (\ell-1)\cdot k$ for some integer $k\geq 0$,
then the $b_j$ can be partitioned into 
$\lfloor(N-(\ell-1)k+1)/\ell \rfloor$ subsets such that the sum of each subset is at least $\ell-1$.
\end{lemma}
Figuratively, the lemma says the following.
There are $N$ bottles of water, such that 
each $(\ell-1)$-tuple of bottles contains at least $\ell-1$ litres.
Some water is spilled out of some of the bottles, such that the total amount spilled out is at most $(\ell-1)\cdot k$ litres.
Then, the bottles can be grouped into $\lfloor(N-(\ell-1)k+1)/\ell \rfloor$ subsets, such that each subset contains at least $\ell-1$ litres of water.
Note that Lemma \ref{lem:remove-goods-b} corresponds to the special case  $\ell=2$.

\begin{proof}[Proof of Lemma \ref{lem:remove-goods-bb}]
Let $a_j := b_j/c_j$ for all $j\in[N]$.
Then $a_j\in[0,1]$,
and 
\begin{align*}
\sum_{j=1}^N a_j 
&= \sum_{j=1}^N \frac{b_j}{c_j}
= N - \sum_{j=1}^N \frac{c_j-b_j}{c_j}.
\end{align*}
Let $d_j := \frac{c_j-b_j}{c_j}$. By assumption,
\begin{align*}
\sum_j c_j\cdot d_j = \sum_j c_j- \sum_j  b_j
\leq (\ell-1)k.
\end{align*}
Applying Lemma \ref{lem:products} in the contrapositive direction with $L := (\ell-1)k$ implies that 
\begin{align*}
\sum_j d_j \leq (\ell-1)k
\end{align*}
too. Therefore,
\begin{align*}
\sum_{j=1}^N a_j  = N - \sum_j d_j \geq N - (\ell-1)k.
\end{align*}
By Lemma \ref{lem:remove-goods-aa},
the $a_j$ can be partitioned into $\lfloor (N-(\ell-1)k+1)/\ell \rfloor$ subsets with a sum of at least $\ell-1$.
In each such subset $G$,
\begin{align*}
\sum_{j\in G}a_j \geq \ell-1.
\end{align*}
Applying Lemma \ref{lem:products} with $L := \ell-1$ and $d_j := a_j$ yields
\begin{align*}
\sum_{j\in G}c_j\cdot a_j \geq \ell-1
\implies 
\sum_{j\in G}b_j \geq \ell-1.
\end{align*}
Therefore, there exists a partition of the $b_j$ as claimed.
\end{proof}

\begin{corollary}
For every $\ell\geq 2$, 
the Lone Divider algorithm can attain \rev{a \tfair{} division} of goods with 
$t_i = \mms{i}{(\ell-1)}{(\ell n -2)}{\objects}$.
\end{corollary}
\begin{proof}
It is sufficient to prove that the $t_i$ are reasonable thresholds for each agent $i$. We verify condition (2) for every integer $k\geq 0$.
Let $N := \ell n-2$, and let $(C_j)_{j\in[N]}$ be a partition of $C$ attaining the maximum in the definition of 
$\mms{i}{\ell-1}{\ell n -2}{\objects}$.
By definition of the MMS, 
the sum of every $(\ell-1)$-tuple of $V_i(C_j)$ is at least $t_i$.
Let $c_j := (\ell-1)\cdot V_i(C_j) / t_i$;
so 
the sum of every $(\ell-1)$-tuple of $c_j$ is at least $\ell-1$.

Suppose we remove some objects whose total value is at most $k\cdot t_i$.
Let $b_j := $ the value of objects remaining at $C_j$, after the removal, multiplied by $(\ell-1) / t_i$;
so $b_j\in[0,c_j]$, and 
$\sum_j b_j \geq (\sum_j c_j) - (\ell-1)\cdot k$.
Lemma \ref{lem:remove-goods-bb} implies that 
the $b_j$ can be partitioned into subsets with a value of at least $\ell-1$.
This corresponds to a partition of the remaining goods into bundles with a value of at least $t_i$.
The number of such bundles is at least

\begin{align*}
\lfloor(N-(\ell-1)k+1)/\ell \rfloor
&=
\left\lfloor \frac{\ell n - 1 - (\ell-1)k}{\ell} \right\rfloor
=
\\
&=
\left\lfloor n - k + \frac{k-1}{\ell} \right\rfloor
\\
&\geq n-k.
\end{align*}
So condition (2) holds, and by Theorem \ref{thm:lone-divider-general}, Lone Divider finds \rev{a \tfair{} division}.
\end{proof}

\subsection{Multiplicative approximation}
The Lone Divider algorithm can also provide a multiplicative approximation, similarly to the ``APX-MMS'' algorithm of \citet{amanatidis2017approximation}. The proof uses the following lemma.

\begin{lemma}
\label{lem:remove2/3}
Let $(b_j)_{j=1}^{N}$, $(c_j)_{j=1}^{N}$  be real numbers 
such that for all $j\in[N]$: 
$b_j\in[0,c_j]$ and $c_j\geq 1$.
If $\sum_j b_j \geq \sum_j c_j - \frac{2n}{3n-1} k$ for some integer $k\geq 0$,
then the $b_j$ can be partitioned into 
$n-k$ subsets such that the sum of each subset is at least $\frac{2n}{3n-1}$.
\end{lemma}
\begin{proof}
Define
\begin{itemize}
\item $d_0$ --- the number of $j$ such that $(c_j-b_j)/c_j\in[0,\frac{n-1}{3n-1}]$.
\item $d_1$ --- the number of $j$ such that $(c_j-b_j)/c_j\in(\frac{n-1}{3n-1},\frac{2n-1}{3n-1}]$.
\item $d_2$ --- the number of $j$ such that $(c_j-b_j)/c_j\in(\frac{2n-1}{3n-1},1]$.
\end{itemize}
We have
\begin{align*}
& d_0+d_1+d_2 = n && 
\\
& \frac{n-1}{3n-1} d_1 + \frac{2n-1}{3n-1} d_2
<
\sum_{j=1}^N (c_j-b_j)/c_j
\leq
\sum_{j=1}^N (c_j-b_j)
&& \text{since $c_j\geq 1$}
\\
&
\leq \frac{2n}{3n-1} k
&& \text{by assumption}
\\
\implies 
& (1/2 - 1/2n)d_1 + (1 - 1/2n)d_2 < k &&  
\\
\implies 
& (d_0+d_1+d_2) - (1/2 - 1/2n)d_1 - (1 - 1/2n)d_2 > n - k &&  
\\
\implies 
& d_1/2 + d_0 + (d_1+d_2)/2n > n - k &&  
\end{align*}
Since $d_1+d_2\leq n$, this implies
\begin{align*}
d_1/2 + d_0 + 1/2 > n - k 
\end{align*}
If $d_1$ is even, then this implies
\begin{align*}
d_1/2 + d_0 \geq  n - k
\end{align*}
If $d_1$ is odd, then the left-hand side is integer, and we get
\begin{align*}
&
d_1/2 + d_0 + 1/2 \geq n - k + 1
\\
\implies &
(d_1-1)/2 + d_0 \geq n - k.
\end{align*}
In both cases 
\begin{align*}
\lfloor d_1/2 \rfloor + d_0 \geq n - k.
\end{align*}
There are $d_0$ indices $j$ for which 
$b_j/c_j \geq  1 - \frac{n-1}{3n-1} = \frac{2n}{3n-1}$, so $b_j\geq \frac{2n}{3n-1}$ too.

There are also $d_1$ indices $j$ for which $b_j/c_j\geq 1 - \frac{2n-1}{3n-1} = \frac{n}{3n-1}$, so $b_j\geq \frac{n}{3n-1}$ too. Pairing these elements gives $\lfloor d_1/2 \rfloor$ pairs with a sum of at least $\frac{2n}{3n-1}$. 
All in all, there are $\lfloor d_1/2 \rfloor + d_0 \geq n - k$ sets
(singletons or pairs) with sum at least $\frac{2n}{3n-1}$.
\end{proof}

\begin{corollary}
The Lone Divider algorithm can attain a fair allocation of discrete goods with $t_i = \frac{2n}{3n-1}\mms{i}{1}{n}{\objects} >  \frac{2}{3}\mms{i}{1}{n}{\objects}$.
\end{corollary}
\begin{proof}
By Theorem \ref{thm:lone-divider-general}, 
it is sufficient to prove that the $t_i$ are reasonable thresholds for each agent $i$.
Let $(C_j)_{j\in[n]}$ be a 1-out-of-$n$ MMS partition of $i$.
Let $c_j := V_i(C_j) / \mms{i}{1}{n}{\objects}$. 
By definition of the MMS, $c_j \geq 1$ for all $j\in[n]$.

For every $k\geq 0$, suppose we remove some objects whose total value is at most $\frac{2n}{3n-1}k\cdot \mms{i}{1}{n}{\objects}$.
Let $b_j := $ the total value remaining in $C_j$, divided by $\mms{i}{1}{n}{\objects}$;
so $b_j\in[0,c_j]$, and the sum of $b_j$ is at least $\sum_j c_j - \frac{2n}{3n-1} k$.
Lemma \ref{lem:remove2/3} implies that the
$b_j$ can be partitioned into $n-k$ subsets with a sum of at least $\frac{2n}{3n-1}$.
This corresponds to a partitioning of the remaining objects into $n-k$ bundles with a value of at least $t_i$. 
Therefore, both conditions (1) and (2) in Definition \ref{def:reasonable} are satisfied.
\end{proof}

\section{Related concepts}
\label{sec:related}

\subsection{Similar concepts with a different name}
Concepts similar to envy-free matching appeared in previous papers related to fair division, but they were ``hidden'' inside proofs of more specific algorithms. One goal of the present paper is to uncover these hidden gems. 

As far as we know, the earliest concept similar to envy-free matching was presented by \citet{Kuhn1967Games}[unnumbered lemma in page 31]. Kuhn presents the lemma in matrix form. In graph terminology, his lemma says that an envy-free matching exists whenever $|X|=|Y|$ and there is a vertex $x\in X$ for whom $N_G(\{x\}) = Y$. This is a special case of our Corollary \ref{thm:sufficient}(b). Kuhn used this lemma in an algorithm for fair cake-cutting, which is now known as the Lone Divider algorithm (see Section \ref{sec:app}).

\citet{procaccia2014fair}[sub.3.1] mention another concept similar to envy-free matching, among proofs of other lemmas related to fair allocation of discrete objects.
They constructed a particular bipartite graph that admits a perfect matching between a subset $X^*\subseteq X$ and a subset $Y^*\subseteq Y$, where there are no edges between $X\sm X^*$ and $Y^*$;
their $X^*$ and $Y^*$ correspond to our $X_L$ and $Y_L$, and the ``no edges'' property corresponds to our Theorem \ref{thm:structure}(a).
Since they were mainly interested in the case of a constant number of players, for which $|X|$ is constant, they did not consider efficient algorithms for computing such a matching.%
\footnote{
This construction did not appear in the journal version \citep{kurokawa2018fair}.
}

Later, 
\citet{amanatidis2017approximation}[lem.4.5]
improved this construction and presented a polynomial-time algorithm for finding a non-empty envy-free matching. Their $X\sm X^+$ and $Y\sm \Gamma(X^+)$ correspond to our $X_L$ and $Y_L$ respectively, and their APX-MMS algorithm corresponds to the Lone Divider algorithm (Algorithm \ref{alg:lone-divider-general}) with threshold values corresponding to $2/3$ fraction of the 1-out-of-$n$ MMS (see Appendix \ref{sec:mms-variants}).

Later, \citet{ghodsi2018fair}[def.3.5] presented a construction that corresponds to our construction as follows:
given a maximum-cardinality matching $M$, 
$F_H(M,\widehat{X})$ corresponds to $Y_L$;~~
$N(F_H(M,\widehat{X}))$ corresponds to $X_L$;%
\footnote{
In fact, as we prove in Theorem \ref{thm:structure}, $Y_L$ and $X_L$ depend only on $G$ and are independent of $M$.
}
~~$\widehat{Y_1}\cup \widehat{Y_2}$ corresponds to $X_S$;~~
and $\widehat{X} \sm F_H(M,\widehat{X})$ corresponds to $Y_S$. Their Lemmas 3.6, 3.7, 3.8 correspond to Theorem \ref{thm:structure} parts (b), (c), (e).

Recently, \citet{bogomolnaia2020guarantees} presented a similar concept which they called a ``proper matching'', for finding a min-max cake allocation when agents have general valuations.

Since all these authors used envy-free matching mainly as an intermediate step in a larger algorithm, they did not consider questions such as the uniqueness of the partition, and did not attempt to find an envy-free matching of maximum cardinality or minimum cost.

Our presentation of envy-free matching as a stand-alone graph-theoretic concept allows us to both simplify old algorithms and design new ones, as illustrated in Sections \ref{sec:app} and \ref{sec:app-objects}. 

\subsection{Different concepts with a similar name}
\label{sub:efm-in-markets}
The term \emph{envy-free matching} is used, in a somewhat more specific sense, in the context of markets, both with and without money.

(1) 
In a market with money, there are several buyers and several goods, and each good may have a price. Given a price-vector, an ``envy-free matching'' is an allocation of bundles to agents in which each agent weakly prefers his bundle over all other bundles, given their respective prices. 
This is a relaxation of a \emph{Walrasian equilibrium}. A Walrasian equilibrium is an envy-free matching in which every item with a positive price is allocated to some agent. In a Walrasian equilibrium, the seller's revenue might be low. This motivates its relaxation to envy-free matching, in which the seller may set reserve-prices (and leave some items with positive price unallocated) in order to increase his expected revenue. See, for example, \citet{Guruswami2005Profitmaximizing,Alaei2012Competitive}.

(2)
In a market without money, there are several people who should be assigned to positions.
For example, several doctors have to be matched for residency in hospitals. 
Each doctor has a preference-relation on hospitals (ranking the hospitals from best to worst), and each hospital has a preference relation on doctors. Each doctor can work in at most one hospital, and each hospital can employ at most a fixed number of doctors (called the capacity of the hospital).     A matching has \emph{justified envy} if there is a doctor $d$ and a hospital $h$, such that $d$ prefers $h$ over his current employer, and $h$ prefers $d$ over one of its current employees.
An ``envy-free matching'' is a matching with no justified envy. 
This is a relaxation of a \emph{stable matching}. A stable matching is an envy-free matching which is also non-wasteful --- 
there is no doctor $d$ and a hospital $h$, such that $d$ prefers $h$ over his current employer and $h$ has some vacant positions \citep{wu2018lattice,yokoi2020envy}.
When the hospitals have, in addition to upper quotas (capacities), also \emph{lower quotas},
a stable matching might not exist. This motivates its relaxation to envy-free matching.

(3)
In contrast, our envy-free matching is an abstract graph-theoretic concept: it is defined for any bipartite graph, and does not require any notion of a price or a ranking.

To differentiate the terms, one can use, for example:

For (1) --- ``price envy-free matching'' or ``market envy-free matching'';

For (2) --- ``no-justified-envy matching'' or ``justified-envy-free matching'';

For (3) --- ``binary envy-free matching'' or ``abstract envy-free matching''.

\clearpage

%\section*{References}
\bibliographystyle{elsarticle-num-names-alpha}
\bibliography{../erel,../elad}

\begin{thebibliography}{41}
\providecommand{\natexlab}[1]{#1}
\providecommand{\url}[1]{\texttt{#1}}
\providecommand{\urlprefix}{URL }
\expandafter\ifx\csname urlstyle\endcsname\relax
  \providecommand{\doi}[1]{doi:\discretionary{}{}{}#1}\else
  \providecommand{\doi}[1]{doi:\discretionary{}{}{}\begingroup
  \urlstyle{rm}\url{#1}\endgroup}\fi
\providecommand{\bibinfo}[2]{#2}

\bibitem[{Alaei et~al.(2012)Alaei, Jain, and Malekian}]{Alaei2012Competitive}
\bibinfo{author}{S.~Alaei}, \bibinfo{author}{K.~Jain},
  \bibinfo{author}{A.~Malekian}, \bibinfo{title}{{Competitive Equilibria in Two
  Sided Matching Markets with Non-transferable Utilities}},
  \bibinfo{note}{arXiv preprint 1006.4696}, \bibinfo{year}{2012}.

\bibitem[{Amanatidis et~al.(2017)Amanatidis, Markakis, Nikzad, and
  Saberi}]{amanatidis2017approximation}
\bibinfo{author}{G.~Amanatidis}, \bibinfo{author}{E.~Markakis},
  \bibinfo{author}{A.~Nikzad}, \bibinfo{author}{A.~Saberi},
  \bibinfo{title}{Approximation algorithms for computing maximin share
  allocations}, \bibinfo{journal}{ACM Transactions on Algorithms (TALG)}
  \bibinfo{volume}{13}~(\bibinfo{number}{4}) (\bibinfo{year}{2017})
  \bibinfo{pages}{52}.

\bibitem[{Avvakumov and Karasev(2020)}]{avvakumov2020equipartition}
\bibinfo{author}{S.~Avvakumov}, \bibinfo{author}{R.~Karasev},
  \bibinfo{title}{Equipartition of a segment}, \bibinfo{journal}{arXiv preprint
  arXiv:2009.09862} .

\bibitem[{Aziz et~al.(2015)Aziz, Gaspers, Mackenzie, and Walsh}]{Aziz2015Fair}
\bibinfo{author}{H.~Aziz}, \bibinfo{author}{S.~Gaspers},
  \bibinfo{author}{S.~Mackenzie}, \bibinfo{author}{T.~Walsh},
  \bibinfo{title}{{Fair assignment of indivisible objects under ordinal
  preferences}}, \bibinfo{journal}{Artificial Intelligence}
  \bibinfo{volume}{227} (\bibinfo{year}{2015}) \bibinfo{pages}{71--92}, ISSN
  \bibinfo{issn}{00043702}, \bibinfo{note}{arXiv preprint 1312.6546}.

\bibitem[{Babaioff et~al.(2017)Babaioff, Nisan, and
  Talgam-Cohen}]{Babaioff2017Competitive}
\bibinfo{author}{M.~Babaioff}, \bibinfo{author}{N.~Nisan},
  \bibinfo{author}{I.~Talgam-Cohen}, \bibinfo{title}{{Competitive Equilibria
  with Indivisible Goods and Generic Budgets}}, \bibinfo{note}{arXiv preprint
  1703.08150 v1 (2017-03)}, \bibinfo{year}{2017}.

\bibitem[{Babaioff et~al.(2019)Babaioff, Nisan, and
  Talgam-Cohen}]{babaioff2019fair}
\bibinfo{author}{M.~Babaioff}, \bibinfo{author}{N.~Nisan},
  \bibinfo{author}{I.~Talgam-Cohen}, \bibinfo{title}{Fair Allocation through
  Competitive Equilibrium from Generic Incomes}, in:
  \bibinfo{booktitle}{Proceedings of the Conference on Fairness,
  Accountability, and Transparency}, \bibinfo{organization}{ACM},
  \bibinfo{pages}{180--180}, \bibinfo{note}{arXiv preprint 1703.08150 v2
  (2018-09)}, \bibinfo{year}{2019}.

\bibitem[{Barman and Krishnamurthy(2017)}]{barman2017approximation}
\bibinfo{author}{S.~Barman}, \bibinfo{author}{S.~K. Krishnamurthy},
  \bibinfo{title}{Approximation algorithms for maximin fair division}, in:
  \bibinfo{booktitle}{Proceedings of the 2017 ACM Conference on Economics and
  Computation}, \bibinfo{organization}{ACM}, \bibinfo{pages}{647--664},
  \bibinfo{year}{2017}.

\bibitem[{Beynier et~al.(2019)Beynier, Chevaleyre, Gourv{\`e}s, Harutyunyan,
  Lesca, Maudet, and Wilczynski}]{beynier2019local}
\bibinfo{author}{A.~Beynier}, \bibinfo{author}{Y.~Chevaleyre},
  \bibinfo{author}{L.~Gourv{\`e}s}, \bibinfo{author}{A.~Harutyunyan},
  \bibinfo{author}{J.~Lesca}, \bibinfo{author}{N.~Maudet},
  \bibinfo{author}{A.~Wilczynski}, \bibinfo{title}{Local envy-freeness in house
  allocation problems}, \bibinfo{journal}{Autonomous Agents and Multi-Agent
  Systems} \bibinfo{volume}{33}~(\bibinfo{number}{5}) (\bibinfo{year}{2019})
  \bibinfo{pages}{591--627}.

\bibitem[{Bogomolnaia et~al.(????)Bogomolnaia, Moulin, and
  Stong}]{bogomolnaia2020guarantees}
\bibinfo{author}{A.~Bogomolnaia}, \bibinfo{author}{H.~Moulin},
  \bibinfo{author}{R.~Stong}, \bibinfo{title}{Guarantees in Fair Division:
  general or monotone preferences}, \bibinfo{journal}{arXiv preprint
  arXiv:1911.10009} .

\bibitem[{Brams and Taylor(1996)}]{Brams1996Fair}
\bibinfo{author}{S.~J. Brams}, \bibinfo{author}{A.~D. Taylor},
  \bibinfo{title}{{Fair Division: From Cake Cutting to Dispute Resolution}},
  \bibinfo{publisher}{Cambridge University Press}, \bibinfo{address}{Cambridge
  UK}, ISBN \bibinfo{isbn}{0521556449}, \bibinfo{year}{1996}.

\bibitem[{Budish(2011)}]{budish2011combinatorial}
\bibinfo{author}{E.~Budish}, \bibinfo{title}{{The Combinatorial Assignment
  Problem: Approximate Competitive Equilibrium from Equal Incomes}},
  \bibinfo{journal}{Journal of Political Economy}
  \bibinfo{volume}{119}~(\bibinfo{number}{6}) (\bibinfo{year}{2011})
  \bibinfo{pages}{1061--1103}, ISSN \bibinfo{issn}{00223808}.

\bibitem[{Ch{\`e}ze(2018)}]{cheze2018don}
\bibinfo{author}{G.~Ch{\`e}ze}, \bibinfo{title}{Don't cry to be the first!
  Symmetric fair division algorithms exist}, \bibinfo{note}{arXiv preprint
  1804.03833}, \bibinfo{year}{2018}.

\bibitem[{Dickerson et~al.(2014)Dickerson, Goldman, Karp, Procaccia, and
  Sandholm}]{Dickerson2014Computational}
\bibinfo{author}{J.~P. Dickerson}, \bibinfo{author}{J.~Goldman},
  \bibinfo{author}{J.~Karp}, \bibinfo{author}{A.~D. Procaccia},
  \bibinfo{author}{T.~Sandholm}, \bibinfo{title}{{The Computational Rise and
  Fall of Fairness}}, in: \bibinfo{booktitle}{Proceedings of the Twenty-Eighth
  AAAI Conference on Artificial Intelligence}, AAAI'14,
  \bibinfo{publisher}{AAAI Press}, \bibinfo{pages}{1405--1411},
  \bibinfo{year}{2014}.

\bibitem[{Elkind et~al.(2021)Elkind, Segal-Halevi, and
  Suksompong}]{elkind2021graphical}
\bibinfo{author}{E.~Elkind}, \bibinfo{author}{E.~Segal-Halevi},
  \bibinfo{author}{W.~Suksompong}, \bibinfo{title}{{Graphical Cake Cutting via
  Maximin Share}}, in: \bibinfo{booktitle}{Proceedings of the 30th
  International Joint Conference on Artificial Intelligence (IJCAI '21)},
  \bibinfo{note}{arXiv preprint 2105.04755}, \bibinfo{year}{2021}.

\bibitem[{Even and Paz(1984)}]{Even1984Note}
\bibinfo{author}{S.~Even}, \bibinfo{author}{A.~Paz}, \bibinfo{title}{{A Note on
  Cake Cutting}}, \bibinfo{journal}{Discrete Applied Mathematics}
  \bibinfo{volume}{7}~(\bibinfo{number}{3}) (\bibinfo{year}{1984})
  \bibinfo{pages}{285--296}, ISSN \bibinfo{issn}{0166218X}.

\bibitem[{Gan et~al.(2019)Gan, Suksompong, and Voudouris}]{gan2019envy}
\bibinfo{author}{J.~Gan}, \bibinfo{author}{W.~Suksompong},
  \bibinfo{author}{A.~A. Voudouris}, \bibinfo{title}{Envy-Freeness in House
  Allocation Problems}, \bibinfo{journal}{Mathematical Social Sciences}
  \bibinfo{volume}{101} (\bibinfo{year}{2019}) \bibinfo{pages}{104--106}.

\bibitem[{Garg et~al.(2019)Garg, McGlaughlin, and Taki}]{garg2019approximating}
\bibinfo{author}{J.~Garg}, \bibinfo{author}{P.~McGlaughlin},
  \bibinfo{author}{S.~Taki}, \bibinfo{title}{Approximating maximin share
  allocations}, \bibinfo{journal}{Open access series in informatics}
  \bibinfo{volume}{69}.

\bibitem[{Ghodsi et~al.(2018)Ghodsi, HajiAghayi, Seddighin, Seddighin, and
  Yami}]{ghodsi2018fair}
\bibinfo{author}{M.~Ghodsi}, \bibinfo{author}{M.~HajiAghayi},
  \bibinfo{author}{M.~Seddighin}, \bibinfo{author}{S.~Seddighin},
  \bibinfo{author}{H.~Yami}, \bibinfo{title}{Fair Allocation of Indivisible
  Goods: Improvements and Generalizations}, in: \bibinfo{booktitle}{Proceedings
  of the 2018 ACM Conference on Economics and Computation},
  \bibinfo{organization}{ACM}, \bibinfo{pages}{539--556}, \bibinfo{note}{arXiv
  preprint 1704.00222}, \bibinfo{year}{2018}.

\bibitem[{Guruswami et~al.(2005)Guruswami, Hartline, Karlin, Kempe, Kenyon, and
  McSherry}]{Guruswami2005Profitmaximizing}
\bibinfo{author}{V.~Guruswami}, \bibinfo{author}{J.~D. Hartline},
  \bibinfo{author}{A.~R. Karlin}, \bibinfo{author}{D.~Kempe},
  \bibinfo{author}{C.~Kenyon}, \bibinfo{author}{F.~McSherry},
  \bibinfo{title}{{On Profit-maximizing Envy-free Pricing}}, in:
  \bibinfo{booktitle}{Proceedings of the Sixteenth Annual ACM-SIAM Symposium on
  Discrete Algorithms}, SODA '05, \bibinfo{publisher}{Society for Industrial
  and Applied Mathematics}, \bibinfo{address}{Philadelphia, PA, USA}, ISBN
  \bibinfo{isbn}{0-89871-585-7}, \bibinfo{pages}{1164--1173},
  \bibinfo{year}{2005}.

\bibitem[{Hall(1935)}]{hall1935representatives}
\bibinfo{author}{P.~Hall}, \bibinfo{title}{On representatives of subsets}, in:
  \bibinfo{booktitle}{Classic Papers in Combinatorics},
  \bibinfo{publisher}{Springer}, \bibinfo{pages}{58--62}, \bibinfo{note}{in
  collection from 2009}, \bibinfo{year}{1935}.

\bibitem[{Hopcroft and Karp(1973)}]{HK73}
\bibinfo{author}{J.~E. Hopcroft}, \bibinfo{author}{R.~M. Karp},
  \bibinfo{title}{An {$n^{5/2}$} algorithm for maximum matchings in bipartite
  graphs}, \bibinfo{journal}{SIAM J. Comput.} \bibinfo{volume}{2}
  (\bibinfo{year}{1973}) \bibinfo{pages}{225--231}, ISSN
  \bibinfo{issn}{0097-5397}, \doi{\bibinfo{doi}{10.1137/0202019}},
  \urlprefix\url{https://doi.org/10.1137/0202019}.

\bibitem[{Hosseini et~al.(2021)Hosseini, Searns, and
  Segal-Halevi}]{hosseini2021ordinal}
\bibinfo{author}{H.~Hosseini}, \bibinfo{author}{A.~Searns},
  \bibinfo{author}{E.~Segal-Halevi}, \bibinfo{title}{Ordinal Maximin Share
  Approximation for Goods}, \bibinfo{journal}{arXiv preprint arXiv:2109.01925}
  .

\bibitem[{Hosseini et~al.(2019)Hosseini, Sikdar, Vaish, Wang, and
  Xia}]{hosseini2019fair}
\bibinfo{author}{H.~Hosseini}, \bibinfo{author}{S.~Sikdar},
  \bibinfo{author}{R.~Vaish}, \bibinfo{author}{J.~Wang},
  \bibinfo{author}{L.~Xia}, \bibinfo{title}{Fair Division through Information
  Withholding}, \bibinfo{note}{arXiv preprint 1907.02583},
  \bibinfo{year}{2019}.

\bibitem[{Huang and Lu(2019)}]{huang2019algorithmic}
\bibinfo{author}{X.~Huang}, \bibinfo{author}{P.~Lu}, \bibinfo{title}{An
  algorithmic framework for approximating maximin share allocation of chores},
  \bibinfo{journal}{arXiv preprint arXiv:1907.04505} .

\bibitem[{Irving et~al.(2006)Irving, Kavitha, Mehlhorn, Michail, and
  Paluch}]{irving2006rank}
\bibinfo{author}{R.~W. Irving}, \bibinfo{author}{T.~Kavitha},
  \bibinfo{author}{K.~Mehlhorn}, \bibinfo{author}{D.~Michail},
  \bibinfo{author}{K.~E. Paluch}, \bibinfo{title}{Rank-maximal matchings},
  \bibinfo{journal}{ACM Transactions on Algorithms (TALG)}
  \bibinfo{volume}{2}~(\bibinfo{number}{4}) (\bibinfo{year}{2006})
  \bibinfo{pages}{602--610}.

\bibitem[{Kamiyama et~al.(2021)Kamiyama, Manurangsi, and
  Suksompong}]{kamiyama2021complexity}
\bibinfo{author}{N.~Kamiyama}, \bibinfo{author}{P.~Manurangsi},
  \bibinfo{author}{W.~Suksompong}, \bibinfo{title}{On the complexity of fair
  house allocation}, \bibinfo{journal}{Operations Research Letters} .

\bibitem[{Kuhn(1967)}]{Kuhn1967Games}
\bibinfo{author}{H.~W. Kuhn}, \bibinfo{title}{{On games of fair division}},
  \bibinfo{publisher}{Princeton University Press}, \bibinfo{pages}{29--37},
  \bibinfo{year}{1967}.

\bibitem[{Kurokawa et~al.(2018)Kurokawa, Procaccia, and
  Wang}]{kurokawa2018fair}
\bibinfo{author}{D.~Kurokawa}, \bibinfo{author}{A.~D. Procaccia},
  \bibinfo{author}{J.~Wang}, \bibinfo{title}{Fair Enough: Guaranteeing
  Approximate Maximin Shares}, \bibinfo{journal}{Journal of the ACM (JACM)}
  \bibinfo{volume}{65}~(\bibinfo{number}{2}) (\bibinfo{year}{2018})
  \bibinfo{pages}{8}.

\bibitem[{Luria(2013)}]{Luria2013EnvyFree}
\bibinfo{author}{Z.~Luria}, \bibinfo{title}{{Envy-free matching}},
  \bibinfo{howpublished}{Mathematics Stack Exchange}, \bibinfo{note}{uRL:
  https://math.stackexchange.com/q/621306/29780 (version: 2013-12-29)},
  \bibinfo{year}{2013}.

\bibitem[{Manabe and Okamoto(2010)}]{Manabe2010MetaEnvyFree}
\bibinfo{author}{Y.~Manabe}, \bibinfo{author}{T.~Okamoto},
  \bibinfo{title}{{Meta-Envy-Free Cake-Cutting Protocols}}, in:
  \bibinfo{editor}{P.~Hlin\v{e}n\'{y}}, \bibinfo{editor}{A.~Ku\v{c}era} (Eds.),
  \bibinfo{booktitle}{Mathematical Foundations of Computer Science 2010}, vol.
  \bibinfo{volume}{6281} of \emph{\bibinfo{series}{Lecture Notes in Computer
  Science}}, chap.~\bibinfo{chapter}{44}, \bibinfo{publisher}{Springer Berlin
  Heidelberg}, \bibinfo{address}{Berlin, Heidelberg}, ISBN
  \bibinfo{isbn}{978-3-642-15154-5}, \bibinfo{pages}{501--512},
  \bibinfo{year}{2010}.

\bibitem[{Manurangsi and Suksompong(2019)}]{manurangsi2019envy}
\bibinfo{author}{P.~Manurangsi}, \bibinfo{author}{W.~Suksompong},
  \bibinfo{title}{When do envy-free allocations exist?}, in:
  \bibinfo{booktitle}{Proceedings of the AAAI Conference on Artificial
  Intelligence}, vol.~\bibinfo{volume}{33}, \bibinfo{pages}{2109--2116},
  \bibinfo{year}{2019}.

\bibitem[{Procaccia and Wang(2014)}]{procaccia2014fair}
\bibinfo{author}{A.~D. Procaccia}, \bibinfo{author}{J.~Wang},
  \bibinfo{title}{{Fair Enough: Guaranteeing Approximate Maximin Shares}}, in:
  \bibinfo{booktitle}{Proc. EC-14}, \bibinfo{publisher}{ACM},
  \bibinfo{address}{New York, NY, USA}, ISBN \bibinfo{isbn}{978-1-4503-2565-3},
  \bibinfo{pages}{675--692}, \bibinfo{year}{2014}.

\bibitem[{Pulleyblank(1995)}]{pulleyblank1995matchings}
\bibinfo{author}{W.~R. Pulleyblank}, \bibinfo{title}{Matchings and extensions},
  \bibinfo{journal}{Handbook of combinatorics} \bibinfo{volume}{1}
  (\bibinfo{year}{1995}) \bibinfo{pages}{179--232}.

\bibitem[{Ramshaw and Tarjan(2012)}]{RT12}
\bibinfo{author}{L.~Ramshaw}, \bibinfo{author}{R.~Tarjan}, \bibinfo{title}{{On
  Minimum-Cost Assignments in Unbalanced Bipartite Graphs}},
  \bibinfo{note}{technical report}, \bibinfo{year}{2012}.

\bibitem[{Robertson and Webb(1998)}]{Robertson1998CakeCutting}
\bibinfo{author}{J.~M. Robertson}, \bibinfo{author}{W.~A. Webb},
  \bibinfo{title}{{Cake-Cutting Algorithms: Be Fair if You Can}},
  \bibinfo{publisher}{A K Peters/CRC Press}, \bibinfo{edition}{first} edn.,
  ISBN \bibinfo{isbn}{1568810768}, \bibinfo{year}{1998}.

\bibitem[{Schreiber et~al.(2018)Schreiber, Korf, and
  Moffitt}]{schreiber2018optimal}
\bibinfo{author}{E.~L. Schreiber}, \bibinfo{author}{R.~E. Korf},
  \bibinfo{author}{M.~D. Moffitt}, \bibinfo{title}{Optimal multi-way number
  partitioning}, \bibinfo{journal}{Journal of the ACM (JACM)}
  \bibinfo{volume}{65}~(\bibinfo{number}{4}) (\bibinfo{year}{2018})
  \bibinfo{pages}{1--61}.

\bibitem[{Segal-Halevi(2021)}]{segalhalevi2020multicake}
\bibinfo{author}{E.~Segal-Halevi}, \bibinfo{title}{Fair Multi-Cake Cutting},
  \bibinfo{journal}{Discrete Applied Mathematics} \bibinfo{volume}{291}
  (\bibinfo{year}{2021}) \bibinfo{pages}{15--35}, \bibinfo{note}{arXiv preprint
  1812.08150}.

\bibitem[{Steinhaus(1948)}]{Steinhaus1948Problem}
\bibinfo{author}{H.~Steinhaus}, \bibinfo{title}{{The problem of fair
  division}}, \bibinfo{journal}{Econometrica}
  \bibinfo{volume}{16}~(\bibinfo{number}{1}) (\bibinfo{year}{1948})
  \bibinfo{pages}{101--104}.

\bibitem[{Woeginger(1997)}]{Woeginger1997Polynomialtime}
\bibinfo{author}{G.~J. Woeginger}, \bibinfo{title}{{A polynomial-time
  approximation scheme for maximizing the minimum machine completion time}},
  \bibinfo{journal}{Operations Research Letters}
  \bibinfo{volume}{20}~(\bibinfo{number}{4}) (\bibinfo{year}{1997})
  \bibinfo{pages}{149--154}, ISSN \bibinfo{issn}{01676377}.

\bibitem[{Wu and Roth(2018)}]{wu2018lattice}
\bibinfo{author}{Q.~Wu}, \bibinfo{author}{A.~E. Roth}, \bibinfo{title}{The
  lattice of envy-free matchings}, \bibinfo{journal}{Games and Economic
  Behavior} \bibinfo{volume}{109} (\bibinfo{year}{2018})
  \bibinfo{pages}{201--211}.

\bibitem[{Yokoi(2020)}]{yokoi2020envy}
\bibinfo{author}{Y.~Yokoi}, \bibinfo{title}{Envy-free matchings with lower
  quotas}, \bibinfo{journal}{Algorithmica}
  \bibinfo{volume}{82}~(\bibinfo{number}{2}) (\bibinfo{year}{2020})
  \bibinfo{pages}{188--211}.

\end{thebibliography}

\end{document}